%% file: main.tex
\documentclass[sigconf]{acmart}

\copyrightyear{2022}
\acmYear{2022}
\setcopyright{acmcopyright}
\acmConference[CCS '22] {Proceedings of the 2022 ACM SIGSAC Conference on Computer and Communications Security}{November 7--11, 2022}{Los Angeles, CA, USA.}
\acmBooktitle{Proceedings of the 2022 ACM SIGSAC Conference on Computer and Communications Security (CCS '22), November 7--11, 2022, Los Angeles, CA, USA}
\acmPrice{15.00}
\acmISBN{978-1-4503-9450-5/22/11}
\acmDOI{10.1145/3548606.3560608}

\input{macro}

\allowdisplaybreaks

\begin{document}
\title{Frequency Estimation in the Shuffle Model with~Almost~a~Single~Message}

\author{Qiyao Luo}
\email{qluoak@cse.ust.hk}
\affiliation{%
  \institution{Hong Kong University of \\ Science and Technology}
  \country{Hong Kong, China}
}

\author{Yilei Wang}
\email{fengmi.wyl@alibaba-inc.com}
\affiliation{%
  \institution{Alibaba Group}
  \country{China}
}

\author{Ke Yi}
\email{yike@cse.ust.hk}
\affiliation{%
  \institution{Hong Kong University of \\ Science and Technology}
  \country{Hong Kong, China}
}

\renewcommand{\shorttitle}{Frequency Estimation in the Shuffle Model with Almost a Single Message}

\begin{abstract}
We present a protocol in the shuffle model of differential privacy (DP) for the \textit{frequency estimation} problem that achieves error $\omega(1)\cdot O(\log n)$, almost matching the central-DP accuracy, with $1+o(1)$ messages per user.  This exhibits a sharp transition phenomenon, as there is a lower bound of $\Omega(n^{1/4})$ if each user is allowed to send only one message.  Previously, such a result is only known when the domain size $B$ is $o(n)$.  For a large domain, we also need an efficient method to identify the \textit{heavy hitters} (i.e., elements that are frequent enough).  For this purpose, we design a shuffle-DP protocol that uses $o(1)$ messages per user and can identify all heavy hitters in time polylogarithmic in $B$.  Finally, by combining our frequency estimation and the heavy hitter detection protocols, we show how to solve the $B$-dimensional \textit{1-sparse vector summation} problem in the high-dimensional setting $B=\Omega(n)$, achieving the optimal central-DP MSE $\softO(n)$  with $1+o(1)$ messages per user. In addition to error and message number, our protocols improve in terms of message size and running time as well.  They are also very easy to implement.  The experimental results demonstrate order-of-magnitude improvement over prior work. 
\end{abstract}



\maketitle

\section{Introduction}

\subsection{Central-DP, local-DP, and shuffle-DP}
A randomized mechanism $\MM:\XX^n \rightarrow \YY$ is \textit{$(\varepsilon, \delta)$-differentially private (DP)} for some privacy parameters $\varepsilon > 0$ and $0 \leq \delta < n^{-\Omega(1)}$, if for any two neighboring datasets $\bm{D} \sim \bm{D}'$ (i.e., $\bm{D}$ and $\bm{D}'$ differ by one element), and any set of outputs $Y \subseteq \YY$,
\begin{equation}
    \label{eq:DP}
\Pr [\MM(\bm{D}) \in Y] \leq e^{\varepsilon} \cdot \Pr [\MM(\bm{D}') \in Y] + \delta.
\end{equation}
Depending on how $\MM(\bm{D})$ is defined, we arrive at different DP models. In \textit{central-DP}, there is a trusted curator who has direct access to the entire dataset $\bm{D}=(x_1,\dots,x_n)$ and $\MM(\bm{D})$ is the privatized estimate of the desired function. In \textit{local-DP}, each user $i$ locally privatizes their own datum $x_i$ using a randomizer $\RR$, and sends $\RR(x_i)$ to an untrusted analyzer $\AAA$.  The DP requirement \eqref{eq:DP} shall hold on these messages, i.e., $\MM(\bm{D}) = (\RR(x_1),\dots, \RR(x_n))$.  Local-DP provides an alternative approach to \textit{secure multi-party computation (MPC)} with the following advantages:
\begin{itemize}
    \item Simplicity: A local-DP protocol, by definition, uses only one round of one-way communications, as opposed to most MPC protocols that require multiple rounds of two-way interactions.  
    \item Efficiency: Local-DP protocols are usually much more efficient than their MPC counterparts in terms of both communication and computation.  Its communication cost is often $\tilde{O}(n)$, as opposed to at least $\Omega(n^2)$ even for the most basic MPC protocol (assuming a constant fraction of colluding parties), such as bit AND \cite{10.5555/3081738.3081753}.  Furthermore, most local-DP protocols use simple arithmetic operations without any expensive cryptography. 
\end{itemize}

These advantages make local-DP a promising direction for large-scale aggregation problems over private data, which has already been adopted by Apple \cite{apple:overview}, Google \cite{google:COVID}, and Microsoft \cite{microsoft:smartnosise}. But it also has clear weaknesses:
\begin{enumerate}
    \item Weaker privacy guarantee: An MPC transcript reveals information about the private input with negligible probability, but in local-DP, $\varepsilon$ is usually set to a constant\footnote{Thus, in the introduction, we state all results for a constant $\varepsilon$, and assume $\Theta(\log{\frac{1}{\delta}})= \Theta(\log{\frac{1}{\beta}}) = \Theta(\log B) =\Theta(\log n)$ to simplify the bounds.  The $\tilde{O}$ notation further suppresses these polylogarithmic factors.} \cite{apple:overview,google:COVID}.
    \item Low accuracy: While an MPC protocol can compute a given functionality with 100\% accuracy, local-DP protocols incur errors at least $\Omega(\sqrt{n})$ even for the most basic problems (e.g., bit counting) and this is inherent \cite{shi2017distributed, chan2012optimal, chen21distinct}.
\end{enumerate}

The first weakness can often be mitigated in the large-scale setting, as concerned users may opt out, which would not affect the aggregation accuracy too much as long as enough users participate.  In addition, technically speaking, the privacy guarantees of MPC and DP are incomparable: An MPC protocol satisfies \eqref{eq:DP} with $\varepsilon=0$ over any two $\bm{D}, \bm{D}'$ having the same function output but not necessarily neighbors.  This offers stronger protection on the transcript, but it does not prevent the output from revealing sensitive information.  Thus, an MPC protocol is often combined with differential privacy to provide all-round protection, i.e., a central-DP mechanism for the function is implemented inside the MPC \cite{pettai2015combining, kairouz2015secure, he2017composing}.  
This makes the strong privacy guarantee of MPC a bit ``excessive'', since the overall privacy guarantee of the whole system is decided by the weakest link, which is differential privacy.

However, the second weakness is still a major concern for local-DP.  In fact, it is due to this reason that the above MPC $+$ central-DP approach is still popular, since a central-DP mechanism has error $\tilde{O}(1)$ for many problems.

Instead of MPC $+$ central-DP, which has good accuracy but high costs, a different approach has been suggested, by assuming that the $\mathcal{R}(x_i)$'s are sent to the analyzer anonymously \cite{Ishai06}.  Combined with differential privacy, this results in the \textit{shuffle-DP} model \cite{badih2021on, borja2019the, ghazi21aggregation, balle20:sum, Bittau2017prochlo, cheu2019distributed, chen21distinct}.  More precisely, shuffle-DP applies the requirement \eqref{eq:DP} on the multiset $\MM(\bm{D})=\{\RR(x_1), \dots, \RR(x_n)\}$, or equivalently, the messages are given to the analyzer after a random shuffle.  Shuffle-DP keeps the simplicity and efficiency of local-DP, modulo the cost of the shuffler, while possibly offering significant improvement in terms of accuracy.  For the bit counting problem, a simple randomized response protocol, in which each user sends a random bit with probability $\tilde{O}(1/n)$, otherwise the true input $x_i$, can already achieve $\tilde{O}(1)$ error in shuffle-DP, matching the error in central-DP up to polylogarithmic factors.

\paragraph{Single-message shuffle-DP}
In the randomized response protocol above, each user only sends one message (actually, just one bit). Moving beyond bit counting, people realized that the single-message shuffle-DP model has its limitation.  Generalizing bit counting to the \textit{real summation} problem, i.e., each $x_i \in  [0,1]$ and we wish to estimate $\sum_i x_i$, we encounter a lower bound of $\Omega(n^{1/6})$ \cite{borja2019the}. For the frequency estimation problem, which is the focus of this paper, there is a lower bound of  $\Omega(n^{1/4})$ \cite{badih2021on}.  On the other hand, central-DP can still achieve error $\tilde{O}(1)$ for both problems.

\paragraph{Multi-message shuffle-DP}
In view of the lower bounds above, it is natural to study shuffle-DP protocols where each user sends multiple messages.  For real summation, the protocol in \cite{Ishai06} achieves $O(1)$ error by sending $O(\log n)$ messages per user.  This was later improved to $O(1)$
messages \cite{balle20:sum}.  Recently, \citet{ghazi21aggregation} have designed a simple and elegant protocol that further reduces the (expected) number of messages per user to $1+o(1)$, and reduces the constant in the $O(1)$ error bound to almost match the error achievable in central-DP.  The reduction in the message number is important both theoretically and practically: Theoretically, it exhibits a sharp transition phenomenon that $o(1)$ extra messages can bring a polynomial improvement in the error; practically, reducing the message number is important since each message has to be anonymized by the shuffler, which could be the dominating cost in a shuffle-DP system, depending on the implementation of the shuffler and its trust assumptions. 

\subsection{Frequency Estimation}
In the frequency estimation problem, each of the $n$ users holds an element from a domain $[B]\coloneqq\{0,\dots,B-1\}$ and the goal is to estimate the number of users holding element $x$ for any $x \in [B]$.  Let $\bm{D}=(x_1, \dots, x_n) \in [B]^n$ be the input, and let $g_x = \sum_{i=1}^n \mathbb{I}[x_i = x]$ be the frequency of $x$ in $\bm{D}$.  A frequency estimation protocol $P:[B]^n \rightarrow \mathbb{N}^B$ returns an estimate of $g_x$ for each $x\in [B]$.  We are interested in the maximum estimation error, i.e., for a failure probability $\beta$, we say that $P$ has error $\alpha$ if
\[\Pr \left[ \max_{x \in [B]} \left|P(\bm{D})_x - g_x \right| > \alpha \right] < \beta. \]

When $B=2$, the problem degenerates into bit counting.  But for the frequency estimation problem, we are more interested in the large-domain case $B\gg n$. For example, when estimating the popularity of hashtags and  websites without any \textit{a priori} restrictions on the possible hashtags/websites, the domain would be all strings (up to a certain length) and all IP addresses. For large $B$, $P(\bm{D})$ shall be returned implicitly, i.e., it is a data structure from which $P(\bm{D})_x$ can be extracted for any given query $x$. 

\citet{badih2021on} prove a lower bound of $\Omega(n^{1/4})$ if each user is only allowed a single message; they also present a protocol that uses $O(\log n)$ messages to achieve $O(\log n)$ error, matching the optimal error in central-DP. In this paper, we present a new shuffle-DP frequency estimation protocol that achieves (i) $O(\log n)$ error with $O(1)$ messages; or (ii) $\omega(1) \cdot O(\log n)$ error with $1+o(1)$ messages, where $\omega(1)$ denotes any super-constant function. Thus, similar to the real summation problem, we see a sharp transition phenomenon for the frequency estimation problem as well.

In addition to error and message complexity, our protocol also  reduces each message size from $O(\log^2 n)$ bits \cite{badih2021on} to $O(\log n)$ bits, and reduce the query time from $O(n\log^3 n)$ \cite{badih2021on} to $O(n)$.  It is also easy to implement, requires only private randomness, and defends against any constant fraction of malicious users (known as \textit{robust shuffle-DP}).
	
\subsection{Heavy Hitter Detection}
For a large domain, estimating the frequency of each element one by one is impractical.  Instead, it should be equipped with a \textit{heavy hitter detection} technique to identify the elements that are frequent enough.  For a threshold $0<\phi < 1$, the set of \textit{heavy hitters} are $\{ x \mid g_x \ge \phi n, x \in [B] \}$.  Note that there are at most $1/\phi$ heavy hitters, thus the goal of heavy hitter detection is to find all of them (their identifiers) in time that depends on $1/\phi$.  Afterwards, one can query a frequency estimation data structure to obtain their frequencies and remove the false positives, subject to an error of $\tilde{O}(1)$.
	
The heavy hitter detection shuffle-DP protocol in \cite{badih2021on} works for any $\phi = \Omega(\log^2 n/ n)$, sends $O(\log^4 n)$ messages per user, and the detection time is $O(n \log^6 n/ \phi)$.  In this paper, we present an improved technique to detect all heavy hitters for any $\phi = \Omega(\log^2 n / n)$ using $O(\log^2 n / \phi n)$ messages and the detection time is $O(\log^2 n / \phi^2)$, both of which are better than \cite{badih2021on} for all $\phi = \Omega(\log^2 n / n)$, and the improvement is more significant for larger $\phi$. In particular, for $\phi=\omega(\log^2 n /n)$, the message number is $o(1)$.  Combined with our $1+o(1)$ frequency estimation protocol, we obtain a $1+o(1)$ protocol to estimate the frequencies of all elements with error $\tilde{O}(1)$ in time $O(n^2)$ for any $B = n^{O(1)}$, where we simply set the frequencies of the light hitters to $0$.

\subsection{1-Sparse Vector Summation}
Finally, we consider the \textit{1-sparse vector summation} problem \cite{ghazi21aggregation}. The $i$-th user holds an element $x_i\in[B]$ together with a weight $w_i \in [0, 1]$, and the goal is to estimate $g_x = \sum_{i=1}^n w_i\cdot \mathbb{I}[x_i = x]$ for all $x \in [B]$. In other words, each user holds a 1-sparse $B$-dimensional vector $\bm{v_i} \in [0, 1]^B$, and the goal is to estimate $\bm{v} = \sum_{i=1}^n \bm{v_i}$.  Note that this problem degenerates into real summation by setting $B=1$, or frequency estimation by setting $w_i =1$ for all $i$.  However, unlike the frequency estimation problem where we are concerned with the maximum error over all elements (the coordinates of $\bm{v}$ in this case), for the vector summation problem, the error metric is often the MSE, i.e., $\E[(\bm{v}-\hat{\bm{v}})^2]$ where $\hat{\bm{v}}$ is an estimated $\bm{v}$.

By directly applying their $1+o(1)$ real summation protocol on each coordinate, \citet{ghazi21aggregation} show how to achieve an MSE of $\tilde{O}(\min\{B,n\})$, matching the optimal central-DP error bound \cite{bun2016simultaneous} with $1+\tilde{O}(\sqrt{B/n} + B/n)$ messages per user\footnote{As stated, their protocol only achieves MSE $\tilde{O}(B)$ and the message number is $1+\softO(B/\sqrt{n})$.  In Section \ref{sec:review}, we show how their protocol can be modified to achieve these improved bounds. }.  This is a $1+o(1)$-message protocol only in the low-dimensional setting $B=o(n)$.
In this paper, we show that in high dimensions $B=\Omega(n)$, this problem can be solved by simply combining our frequency estimation and heavy hitter detection protocols, which can achieve $\tilde{O}(n)$ MSE with $1+o(1)$ messages.  Thus, together with \cite{ghazi21aggregation}, we have shown that the optimal central-DP MSE (up to polylogarithmic factors) can be achieved with $1+o(1)$ messages in shuffle-DP.
	

\section{Related Work}

The idea of anonymizing/shuffling the messages before handing them to an analyzer goes back to at least 1981 \cite{chaum1982untraceable}.  Combining this notion with differential privacy results in the shuffle-DP model, which has attracted much attention in recent years.  In addition to the problems already covered earlier, \citet{chen21distinct} study the distinct count problem in shuffle-DP, which also exhibits a separation between the single-message and multi-message setting.  They prove an error lower bound of $\Omega(n)$ (hence, hopeless) for the former, and an upper bound of $\softO(\sqrt{n})$ for the latter.  The standard shuffle-DP model, like local-DP, only allows a single round of communication.  Recently, \citet{amos2020on} extend this model and investigate what can be achieved with two rounds.  \citet{huang2021instance} design a 3-round mean estimation protocol that achieves instance optimality.  In this paper, we use the standard 1-round model of shuffle-DP.

The protocol in \cite{badih2021on} is the state of the art for frequency estimation and heavy hitter detection for a large $B$, i.e., the dependency on $B$ is logarithmic.  Prior to that, there are a number of protocols with a linear dependency on $B$. \citet{balcer2019separating} provide a protocol that sends $O(B)$ messages per user while achieving $O(\log n)$ error, which keeps independent of the domain size even when $\log B = \Omega(\log n)$.
\citet{albert2021privatehistogram} design a protocol that can send 2 messages per user, but each message has $O(B)$ bits. 
\citet{ghazi2020privatecounting} show how to achieve $O(\log n)$ error with $1 + \softO(B/n)$ messages each of $O(\log B)$ bits.  In Section \ref{sec:smallB} we show how our framework also can recover this result easily.  
Frequency estimation and heavy hitter detection have also been studied extensively in local-DP \cite{wang2017locally, erlingsson2014rappor,apple:overview,bassily2015local,bassily2017practically,feldman2021lossless,wennan2020federated}, but the optimal error in local-DP is $\tilde{\Theta}(\sqrt{n})$.


Compared with local-DP, shuffle-DP allows us to achieve a much smaller error, but it relies on the availability of an anonymizer or shuffler.  From its early conception \cite{chaum1982untraceable}, many practical implementations of shufflers based on different technologies have been released, including mix networks \cite{chaum1982untraceable, danezis2003Mixminion}, onion routing \cite{reed1998onion, roger2004tor}, trusted nodes/hardware \cite{reiter1998crowds, Bittau2017prochlo}, etc.

\section{Frequency Estimation}
	
\subsection{A Balls-into-bins mechanism} \label{sec:bib}
We first study the privacy of a simple balls-into-bins mechanism $\MM^\BIB$, shown in Algorithm~\ref{algo:bib}, in the central-DP model.  This mechanism, as we will see shortly, is equivalent to our shuffle-DP mechanisms with appropriate settings of its parameters $m,s,k,n,p$.  The input to $\MM^\BIB$ is some $S\subseteq [m]$ such that $|S|=s$, and any pair of inputs $S$ and $S'$ are considered neighbors.  The set $S$ can be thought of as $s$  special bins, chosen out of a total of $m$ bins.  The mechanism throws one \textit{real} ball and (expected) $k+np$ \textit{noisy} balls into the $m$ bins, and the numbers of balls in all the bins are taken as the output.  Specifically, it throws the real ball into one of the $s$ special bins uniformly at random.  Then, it throws $k$ noisy balls uniformly at random into all the bins. Finally, it flips a biased coin $n$ times (we use $\Ber(p)$ to denote a Bernoulli random variable with parameter $p$ in the algorithm), and for each heads we get, it throws a noisy ball into one bin chosen uniformly at random from all bins. All the balls are thrown independently. Clearly, for any $S\ne S'$, the distributions of the noisy balls are identical, whose purpose is to hide the location of the real ball.  Since the location of the real ball is randomly chosen from the $s$ special bins, a larger $s$ makes it easier to hide (in the extreme case $s=m$, no noisy ball is needed).  The theorem below formalizes this intuition:
	
	\begin{algorithm}
		\caption{Balls-into-bins Mechanism $\MM^{\BIB}$}	\label{algo:bib}
		\KwPublic{$m,s,k,n,p$}
		\KwIn{A set $S \subseteq [m]$ such that $|S| = s$}
		\KwOut{A multiset $\OO \subseteq [m]$}
		Choose $x \in S$ uniformly at random\;
		$\OO \leftarrow \{x\}$\;
		\For{$i\gets1$ \KwTo $k$}{ 
			Choose $x \in [m]$ uniformly at random\; 
			$\OO \leftarrow \OO \uplus \{x\}$ \tcp*{Increase the multiplicity of $x$ by 1 ($\uplus$  stands for the union operation on multisets).}
		}
		\For{$i\gets1$ \KwTo $n$}{
			$y \gets \Ber(p)$\;
			\If{$y=1$}{ 
				Choose $x \in [m]$ uniformly  at random\;
				$\OO \leftarrow \OO \uplus \{x\}$\;
			}
		}
		\KwRet{$\OO$}\;
	\end{algorithm}
	
\begin{theorem}
	\label{thm:bibdp}
	For any $0 < \varepsilon \leq 3$ and $0 < \delta < 1$, the mechanism $\MM^{\BIB}$ is $(\varepsilon, \delta)$-differentially private if $k+np \ge \frac{32\ln(2/\delta)}{\varepsilon^2}\cdot\frac{m}{s}$. 
\end{theorem}
	
\begin{proof}
The proof follows the framework of Lemma~4.4 in \cite{badih2021on}. Treating the numbers of balls in the $m$ bins as an $m$-dimensional vector, we use $\MM(S) \in \mathbb{N}^m$ to denote the output of $\MM^\BIB$ with public parameters $m,s,k,n,p$ on input $S$.
We prove the following inequality, which is equivalent to the DP definition in \eqref{eq:DP}:
\begin{equation}
	\label{dpdef}
	\Pr_{\bm{W} \sim{\MM}(S) } \left[ \frac{\Pr[{\MM}(S) = \bm{W}]}{\Pr[{\MM}(S') = \bm{W}]} \geq e^{\varepsilon} \right] \leq \delta.
\end{equation}
		
We first consider mechanism $\MM_0^\BIB$ that only throws noisy balls (i.e., skip the first two lines in Algorithm \ref{algo:bib}). Let $\MM_0\in \mathbb{N}^m$ be the output vector of $\MM_0^\BIB$; note that $\MM_0$ does not depend on $S$.
For any $\bm{W}=(w_i)_{i=1}^m\in\mathbb{N}^m$, denote $w=\sum_{i=1}^m w_i$. Then when $k\leq w\leq k+n$, we have
		\[\Pr[\MM_0=\bm{W}]=\binom{n}{w-k}p^{w-k}(1-p)^{n-w+k}\frac{w!}{\prod_{i=1}^m w_i!}m^{-w}.\]
		
Now consider $\MM^\BIB$. Let $I$ be the bin chosen by the real ball (line 1 in Algorithm~\ref{algo:bib}) and $\bm{e_i}$ be a length-$m$ one-hot vector where the $i$-th coordinator is 1 and other coordinators are all 0. 
For any $\bm{W}=(w_i)_{i=1}^m\in\mathbb{N}^m$, denote $w=\sum_{i=1}^m w_i$. Then when $1+k\leq w\leq 1+k+n$, we have
\begin{align*}
	&\Pr[\MM(S)=\bm{W}]\\
	=&\sum_{i\in S}\Pr[\MM(S)=\bm{W}\mid I=i]\cdot\Pr[I=i]\\
	=&\sum_{i\in S} \Pr[\MM_0=\bm{W}-\bm{e_i}]\cdot s^{-1}\\
	=&\binom{n}{w-1-k}p^{w-1-k}(1-p)^{n-w+1+k}\cdot\frac{(w-1)!}{s\prod_{i=1}^m w_i!} \sum_{i\in S}w_i .
\end{align*}
Therefore, inequality~\eqref{dpdef} is equivalent to 
\[
\Pr _ {\bm{W} \sim{\MM}(S) } \left[ \frac{\sum_{i \in S} w_i}{\sum_{i \in S'} w_i} \geq e^{\varepsilon} \right] \leq \delta.
\]

We use $\Bin(n,p)$ to denote a binomial distribution with parameters $n$ and $p$.		
	Notice that when $\bm{W}\sim \MM(S)$, $\sum_{i \in S} w_i\sim 1+\Bin(k,\frac{s}{m})+\Bin(n,\frac{ps}{m})$, while $\sum_{i \in S'} w_i\sim \Ber(|S\cap S'|/s)+\Bin(k,\frac{s}{m})+\Bin(n, \frac{ps}{m})$. Therefore, it suffices to prove
	\begin{equation}
		\label{binodp}
		\Pr\left[ \frac{1 + X_1}{X_2} \geq e^{\varepsilon} \right] \leq \delta,
	\end{equation}
	where $X_1, X_2 \sim \Bin(k,\frac{s}{m})+\Bin(n, \frac{ps}{m})$. Note that $X_1$ and $X_2$ are not necessarily independent.
	
	Since $\mu\coloneqq \E[X_1]=(k+np)\frac{s}{m}\geq\frac{32\ln(2/\delta)}{\varepsilon^2}$, by Chernoff bound,	we have 
	$$\Pr[X_1 \ge (1+\varepsilon/3)\mu] \le  \exp\left(-\frac{\mu}{3} \left(\frac{\varepsilon}{3}\right)^2\right) < \frac{\delta}{2},	$$
	$$\Pr[X_2 \le (1-\varepsilon/4)\mu] \le  \exp\left(-\frac{\mu}{2} \left(\frac{\varepsilon}{4}\right)^2\right) < \frac{\delta}{2}.	$$
	
	Since $\delta < 1$, we have $\mu > (32\ln 2)/\varepsilon^2>18/\varepsilon^2$. Also note that $e^{\varepsilon/3} \geq 1+\varepsilon/3+\varepsilon^2/18$ and $1-\varepsilon/4 > e^{-\varepsilon/2}$ for $0<\varepsilon \leq 3$, so we have
	$$\frac{1+(1+\varepsilon/3)\mu}{(1-\varepsilon/4)\mu}  \leq \frac{\varepsilon^2/18+(1+\varepsilon/3)}{1-\varepsilon/4}\leq \frac{e^{\varepsilon/3}}{e^{-\varepsilon/2}}\leq e^{\varepsilon}.$$
	Finally, by a union bound, we obtain
\[			\Pr\left[\frac{1 + X_1}{X_2} \ge e^{\varepsilon} \right] \\
		\le  \Pr [X_1 \ge  (1+\varepsilon/3)\mu] + \Pr [X_2 \le (1-\varepsilon/4)\mu] \le \delta. \qedhere \]
\end{proof}

\paragraph{Remark}
From the proof of Theorem~\ref{thm:bibdp}, we see that the derivation from 
\begin{equation}\label{eq:old}
    (k+np)\cdot\frac{s}{m} \ge \theta := \frac{32\ln(2/\delta)}{\varepsilon^2}
\end{equation}
to the key inequality~\eqref{binodp} is quite loose. In particular, we chose a large constant $32$ so as to simplify the proof.  In practice, when concrete values of $n,m,s,\varepsilon,\delta$ are given, we can try minimizing $\theta$ as long as \eqref{binodp} holds.  
Specifically, we consider the worst case $S\cap S'=\varnothing$, in which case $\Pr[(1+X_1)/X_2\ge e^\varepsilon]$ is a function of $\theta$ that can be computed in $O(n)$ time.  Then, we perform a binary search for the smallest $\theta$ such that \eqref{binodp} holds.  We stop the binary search when the range has narrowed down to $1/n$, which gives us a near-optimal $\theta$ in $O(n \log n)$ time.  Note that this optimization step is data-independent, so $\theta$ can be pre-computed. Furthermore, it works without the technical condition $0<\varepsilon\le 3$ in  Theorem~\ref{thm:bibdp}.

\subsection{The Protocol for a Small Domain}
\label{sec:smallB}
As a warm-up, we first describe a protocol for a small domain $B<\tilde{O}(n)$.  It is a simple modification of the single-message shuffle protocol of \cite{borja2019the}. In their local randomizer, each user sends its input with probability $\rho=\softO(B/n)$, and sends a \textit{blanket noise} otherwise. The blanket noise is a random element uniformly chosen from $[B]$. When $B=\sqrt{n}$, their analyzer incurs $\softO(n^{1/4})$ error, matching the lower bound of single-message protocols \cite{badih2021on}. 
	
We extend their result to a $(1+\rho)$-message protocol $\PP^{\FE0}=(\RR^{\FE0},\AAA^{\FE0})$ for better accuracy, where $\rho=\softO(B/n)$. The local randomizer works as follows: Each user always sends its input; in addition, with probability $\rho$, it also sends a blanket noise (see Algorithm~\ref{algo:lr-fe0}).

	\begin{algorithm}
		\caption{Local Randomizer $\RR^{\FE0}$}\label{algo:lr-fe0}
		\KwPublic{$B, n, \varepsilon, \delta$}
		\KwIn{$x \in [B]$}
		\KwOut{A multiset $\TT \subseteq [B]$}
		$\TT \gets \{x\}$\;
		$\rho\gets\frac{32 \ln(2/\delta)}{\varepsilon^2} \cdot \frac{B}{n}$\;
		$y\gets\Ber(\rho)$\;
		\If{$y=1$}{
			Choose $x'$ from $[B]$ uniformly\;
			$\TT \gets \TT \uplus \{x'\}$\;
		}
		\KwRet{$\TT$}\;
	\end{algorithm}
	
We first show that this local randomizer satisfies DP:

	\begin{lemma}[Privacy of $\RR^{\FE0}$]
		\label{lem:prife0}
		For any $n, B \in \mathbb{N}$, $0 < \varepsilon \leq 3$, and $0 < \delta < 1$, let $\rho = \frac{32 \ln(2/\delta)}{\varepsilon^2} \cdot \frac{B}{n}$. If $\rho\leq 1$, then $\RR^{\FE0}$ satisfies $(\varepsilon, \delta)$-shuffle DP.
	\end{lemma}
	
	\begin{proof}
		For any input $\bm{D} = (x_1, x_2, \dots, x_n)$, the view of the analyzer is the multiset $\TT$ which can be split into two parts: 
		\begin{enumerate}
			\item $n$ real inputs from all the users $\{x_i\}_{i=1}^n$, and
			\item the multiset of expected $n\rho$ blanket noises $\TT_{\Noise}\subseteq[B]$.
		\end{enumerate}
		Let $\bm{D}'$ be any neighbouring databases of $\bm{D}$. Then we need to prove:
		\begin{equation}
			\label{eq:fe0-origin}
			\Pr_{\bm{\tau}\sim\TT(\bm{D}) } \left[ \frac{\Pr[\TT(\bm{D}) = \bm{\tau}]}{\Pr[\TT(\bm{D'}) = \bm{\tau}]} \geq e^{\varepsilon} \right] \leq \delta.
		\end{equation}
		
		Without loss of generality, we assume $\bm{D}$ and $\bm{D}'$ differ in the last element, i.e., $\bm{D}' = (x_1, x_2, \dots, x_n')$. 
		By factoring out the first $n-1$ users, it suffices to prove:
		\begin{equation}
			\label{eq:fe0-target}
			\Pr_{\bm{\tau}\sim\TT_\Noise } \left[ \frac{\Pr[\TT_\Noise = \bm{\tau}]}{\Pr[\TT_\Noise\uplus\{x'_n\} = \bm{\tau}\uplus \{x_n\}]} \geq e^{\varepsilon} \right] \leq \delta.
		\end{equation}
		Note that \eqref{eq:fe0-target} is actually stronger than \eqref{eq:fe0-origin}, as the first $n-1$ common users provide additional randomness.  To prove \eqref{eq:fe0-target}, we note that without real inputs from the first $n-1$ users, the shuffled messages correspond to $\MM^\BIB$ with public parameters $m=B,s=1,k=0,n,p=\rho$:
		\begin{itemize}
			\item The set of bins in $\MM^\BIB$ corresponds to $[B]$.
			\item The input set $S$ in $\MM^\BIB$ corresponds to $\{x_n\}$ and $\{x'_n\}$, respectively. 
			\item No fixed $k$ noisy balls, i.e., $k=0$.
			\item The multiset of bins chosen by the expected $np$ noisy balls are $\TT_\Noise$.
		\end{itemize} 
		Then the lemma immediately follows from Theorem~\ref{thm:bibdp}.
	\end{proof}

For a query $x$, the analyzer (Algorithm \ref{algo:a-fe0}) simply counts the number of messages that equal to $x$, followed by a bias-removal step.  

\begin{algorithm}
	\caption{Analyzer $\AAA^{\FE0}$}\label{algo:a-fe0}
	\KwPublic{$B, n, \varepsilon, \delta$}
	\KwIn{A multiset $\TT\subseteq[B]$; element $x$}
	\KwOut{Estimated frequency of $x$}
	$X \gets $ the frequency of $x$ in $\TT$\;
	$\rho\gets\frac{32 \ln(2/\delta)}{\varepsilon^2} \cdot \frac{B}{n}$\;
	$\hat{g}_x \gets X - n\rho/B$\;
	\KwRet{$\hat{g}_x$}\;
\end{algorithm}

\begin{lemma}[Accuracy of $\PP^{\FE0}$]
	\label{lem:accfe0}
	The protocol $\PP^{\FE0}$ has error \[\alpha=\max \left\{3\ln(2B/\beta), \sqrt{3\ln(2B/\beta) \cdot \frac{32\ln(2/\delta)}{\varepsilon^2}}\right\}.\]
\end{lemma}
	
\begin{proof}
	Let $X$ be the frequency of $x$ in the output multiset $\TT$, then $X=g_x+\Bin(n,\rho/B)$, where the second term comes from the blanket noises. Since $\hat{g}_x=X-n\rho/B$, we conclude $\E[\hat{g}_x]=g_x$, so $\AAA^{\FE0}$ provides an unbiased estimation.
	
	Let $Y=\Bin(n,\rho/B)$. Define $\mu\coloneqq\E[Y]$, then by Chernoff bound, for $0 < \eta \leq 1$, 
	\[
	\Pr\left[ |Y - \mu| > \eta\mu\right] < 2 \exp(-\frac{\eta^2\mu}{3}).
	\]
	And for $\eta > 1$,
	\[
	\Pr\left[ |Y - \mu| > \eta\mu\right] < 2 \exp(-\frac{\eta\mu}{3}).
	\]
	If $\mu \geq 3\ln(2B/\beta)$, by setting $\eta = \sqrt{\frac{3\ln(2B/\beta)}{\mu}} \leq 1$ we get
	$$\Pr\left[ |Y - \mu| > \sqrt{3\ln(2B/\beta)\mu}\right) < 2 \exp(-\frac{\eta^2\mu}{3}) = \frac{\beta}{B}.$$
	Otherwise, by setting $\eta = \frac{3\ln(2B/\beta)}{\mu} > 1$ we get
	$$\Pr \left[ |Y - \mu| > 3\ln(2B/\beta)\right) < 2 \exp(-\frac{\eta\mu}{3}) = \frac{\beta}{B}.$$
	In conclusion, 
	\[
	\Pr \left[ |Y - \mu| > \max \left\{3\ln(2B/\beta), \sqrt{3\ln(2B/\beta) \cdot \mu} \right \} \right] < \frac{\beta}{B}.
	\]
	Therefore, 
	\[
	\Pr \left[ \left|\hat{g}_x - g_x\right| > \alpha \right] < \frac{\beta}{B}.
	\]
	By union bound, 
	\[
	\Pr \left[ \bigvee_{x\in [B]}|\hat{g}_x - g_x| > \alpha \right] < \beta,
	\]
	which implies the result.
\end{proof}

The number of messages sent per user, message size, and query time are obvious. Thus we have:
\begin{theorem}
For $0 < \varepsilon \leq 3$, $0 < \delta < 1$, and any $n, B \in \mathbb{N}$ such that $B\leq\frac{\varepsilon^2 n}{32\ln(2/\delta)}$, $\PP^{\FE0}$ is a private-coin $(\varepsilon, \delta)$-shuffle DP frequency estimation protocol  that sends $1 + O\left(\frac{\log(1/\delta)}{\varepsilon^2} \cdot \frac{B}{n}\right)$ messages per user in expectation, each consisting of $O(\log B)$ bits. Any frequency query can be answered in expected $O(n)$ time with error $O\left(\log{\frac{B}{\beta}} + \frac{1}{ \varepsilon} \sqrt{\log{\frac{B}{\beta}} \log(1/\delta)}\right)$.
\label{the:fe0}
\end{theorem}
	
Note that for $B = o\left(\frac{\varepsilon^2 n}{ \log(1/\delta)}\right)$, $\PP^{\FE0}$ is a $(1+o(1))$-message protocol with error $\tilde{O}(1)$.

\subsection{The Protocol for a Large Domain}
The protocol above can be easily extended to support a larger domain.  For $B\ge \tilde{\Omega}(n)$, $\rho$ will be greater than $1$, but we can ask each user to first send $\lfloor \rho \rfloor$ blanket noises, and then another blanket noise with probability $\rho-\lfloor \rho \rfloor$.  Privacy still holds, by appropriately modifying the proof of Lemma \ref{lem:prife0}, but the number of messages will be $1+\rho = \tilde{O}(B/n)$.   In this section, we provide another protocol $\PP^{\FE1}$, which sends $1+o(1)$ messages for an arbitrarily large $B$ while achieving $\softO(1)$ error.
	
The idea is to use hashing to reduce the domain size. Specifically, we use the classical hash function $h_{u,v}:~[B]\to [b]$ to hash an element in $[B]$ to a bin in $[b]$, where $h_{u,v}(x)=((ux+v)\bmod{q})\bmod{b}$, $q\in[B,2B)$ is a prime\footnote{This $q$ must exist due to  Bertrand's postulate.}, and $b$ is a parameter that will control the trade-off between communication cost and accuracy.
Let $\HH = \{h_{u,v}\mid (u,v)\in\{1,\dots,q-1\}\times[q]\}$. It is well-known that $\HH$ is a universal family with collision probability $\pcol(x,y)\coloneqq\Pr[h_{u,v}(x)=h_{u,v}(y)]<\frac{1}{b}$. In fact, for this family, the collision probability is the same for all $x\ne y$: $\pcol(\cdot,\cdot) \equiv \lfloor q/b \rfloor((q\bmod{b})+q-b)/(q(q-1))$, which is simply denoted as $\pcol$.
	
\subsubsection{Local randomizer (Algorithm~\ref{algo:lr-fe1})}
Assume a user holds a private input $x$. It uniformly randomly chooses $(u,v)\in\{1,\dots,q-1\}\times[q]$, and sends $(u,v, h_{u,v}(x))$ to the shuffler, called the \textit{real output}. 
It also sends $\rho$ \textit{blanket noises} to the shuffler, where $\rho=\frac{32 \ln(2/\delta)}{\varepsilon^2} \cdot \frac{b}{n}$.  
Each blanket noise $(u,v,w)$ is uniformly chosen from $\{1,\dots,q-1\}\times[q]\times[b]$ independently. Note that when $\rho$ is not an integer, it sends $\lfloor \rho \rfloor$ blanket noises and then with probability $\rho-\lfloor \rho \rfloor$ sends another blanket noise, so that the expected number of blanket noises is $\rho$. The total expected number of messages the user sends is therefore $1+\rho$, including one real output and $\rho$ blanket noises.
	
	\begin{algorithm}
		\caption{Local Randomizer $\RR^{\FE1}$}\label{algo:lr-fe1}
		\KwPublic{$B, b, n, \varepsilon, \delta$}
		\KwIn{$x \in [B]$}
		\KwOut{A multiset $\TT \subseteq \{1,\dots,q-1\} \times [q] \times [b]$}
		Choose $(u,v)$ from $\{1,\dots,q-1\}\times[q]$ uniformly\;
		$\TT \gets \{(u,v, h_{u,v}(x))\}$\;
		$\rho\gets\frac{32 \ln(2/\delta)}{\varepsilon^2} \cdot \frac{b}{n}$\;
		\For{$i\gets 1$ \KwTo $\lfloor \rho \rfloor$}{
			Choose $(u,v,w)$ from $\{1,\dots,q-1\}\times[q]\times[b]$ uniformly\;
			$\TT \gets \TT \uplus \{(u,v,w)\}$\;
		}
		$y \gets \Ber(\rho-\lfloor \rho \rfloor)$\; 
		\If{$y=1$}{
			Choose $(u,v,w)$ from $\{1,\dots,q-1\}\times[q]\times[b]$ uniformly\;
			$\TT \gets \TT \uplus  \{(u,v,w)\}$\;
		}
		\KwRet{$\TT$}\;
	\end{algorithm}
	
	\begin{lemma}[Privacy of $\RR^{\FE1}$]
		\label{lem:prife1}
		For any $n, B, b \in \mathbb{N}$, $0 < \varepsilon \leq 3$, and $0 < \delta < 1$, let $\rho = \frac{32 \ln(2/\delta)}{\varepsilon^2} \cdot \frac{b}{n}$. Then $\RR^{\FE1}$ satisfies $(\varepsilon, \delta)$-shuffle DP.
	\end{lemma}
		\begin{proof}
		The proof is similar to the proof of Lemma~\ref{lem:prife0}: First we exclude $\{(u_i,v_i,h_{u_i,v_i}(x_i))\}_{i=1}^{n-1}$, i.e., the real outputs from the first $n-1$ users. Then the privacy follows from that of $\MM^\BIB$ with parameters $m=(q-1)qb,s=(q-1)q,k=n\lfloor \rho \rfloor,n,p=\rho-\lfloor \rho \rfloor$:
			\begin{itemize}
			\item The set of bins in $\MM^\BIB$ are $\{1,\dots,q-1\} \times [q] \times [b]$.
			\item The input set $S$ in $\MM^\BIB$ corresponds to $\{(u,v,w) \mid h_{u,v}(w)=x_n\}$ and $\{(u,v,w) \mid h_{u,v}(w)=x'_n\}$, respectively. 
			\item Each user first throws $\lfloor \rho \rfloor$ noisy balls, so $k=n\lfloor \rho \rfloor$.
			\item Each user then throws a noisy ball with probability $p=\rho-\lfloor \rho \rfloor$. \qedhere
		\end{itemize}
	\end{proof}
	
	\begin{table*}
		\begin{center}
			\renewcommand\arraystretch{1.5}
			\begin{tabular}{c | c | c | c | c | c | c  } 
				\hline
				\multicolumn{2}{c|}{Protocol} & Messages per user & Message size  & Error & Query time (single) & Query time (all) \\ [0.5ex] 
				\hline\hline
				\multicolumn{2}{c|}{} & 2 & $B$ & $O(\log n)$ & $O(n)$ & $O(Bn)$ \\ 
				\cline{3-7}
				\multicolumn{2}{c|}{CZ \cite{albert2021privatehistogram}} & 2 & $O(\log n + B \log^2n / n)$ & $O(\log n)$ & $O(n \log n)$ & $O(n \log n + B \log^2 n)$ \\ 
				\cline{3-7}
				\multicolumn{2}{c|}{} & $O(\log n)$ & $O(\log^5 n)$ & $O(\log^2 n)$ & $O(n\log n)$ & $O(n\log n + B\log n)$ \\ 
				\hline
				\multicolumn{2}{c|}{GGKPV \cite{badih2021on}}   & $O(\log n)$ & $O(\log^2 n)$ & $O(\log n)$ & $O(n \log^3 n)$ & $O(n \log^4 n+B\log^2 n)$\\ 
				\hline
				\multicolumn{2}{c|}{This work (small $B$)} & $1 + O(B\log n/n)$  & $O(\log n)$ & $O(\log n)$ & $O(n)$ & $O(n)$ \\ 
				\hline
				\multicolumn{2}{c|}{This work (large $B$)} & $1 + O(b\log n/n)$ & $O(\log n)$ & $O(\log n+\sqrt{\frac{n}{b}\log n})$ & $O(n+b\log n)$ & $O(Bn/b+B\log n)$\\ 
				\hline
				Above with  & $c=1$ & $O(1)$ & $O(\log n)$ & $O(\log n)$ & $O(n)$ &$O(B\log n)$ \\
				\cline{2-7}
				$b=\frac{n}{\log^c n}$ & $c>1$ & $1+o(1)$ & $O(\log n)$ & $O(\log^{\frac{c+1}{2}} n)$ & $O(n)$ & $O(B\log^c n)$ \\ 
				\hline	\multicolumn{2}{c|}{Corollary~\ref{col:hhdforfe} ($c>2$)} & $1 + o(1)$  & $O(\log n)$ & $O(\log^c n)$ & - & $O(n^2/\log^c n)$ \\ 
				\hline
			\end{tabular}
		\end{center}
		\caption{Comparison with prior works for frequency estimation, assuming $\varepsilon=\Theta(1)$ and $\log n = \Theta(\log(1/\delta)) =\Theta(\log B) = \Theta(\log(1/\beta))$.}
		\label{tab:fecomp}
	\end{table*}
	
	\subsubsection{Analyzer (Algorithm~\ref{algo:a-fe1})} The analyzer's view is a multiset of tuples $(u,v,w)$. To query for the frequency of an element $x$, it scans all tuples received and count the number of tuples satisfying $h_{u,v}(x) =w$. Then we remove the bias caused by hash collisions and the blanket noises.
	
	\begin{algorithm}
		\caption{Analyzer $\AAA^{\FE1}$}\label{algo:a-fe1}
		\KwPublic{$B, b, n, \varepsilon, \delta$}
		\KwIn{An element $x\in[B]$; a multiset $\TT\subseteq\{1,\dots,q-1\}\times[q]\times[b]$}
		\KwOut{Estimated frequency of $x$}
		$\pcol\gets \lfloor q/b \rfloor((q\bmod{b})+q-b)/(q(q-1))$\;
		$X \gets 0$\;
		\For{$(u,v,w) \in \TT$}{ \If{$h_{u,v}(x) = w$}{$X \gets X + 1$}}
		$\rho\gets\frac{32 \ln(2/\delta)}{\varepsilon^2} \cdot \frac{b}{n}$\;
		$\hat{g}_x \gets (X - n\rho/b-n\pcol)/(1-\pcol)$\;
		\KwRet{$\hat{g}_x$}\;
	\end{algorithm}
	
	\begin{lemma}[Accuracy of $\PP^{\FE1}$]
		\label{lem:accfe1}
		The protocol $\PP^{\FE1}$ has error \[\alpha=2\max \left\{3\ln(2B/\beta), \sqrt{3\ln(2B/\beta) \cdot  \left(\frac{n}{b} + \frac{32\ln(2/\delta)}{\varepsilon^2}\right)}\right\}.\]
	\end{lemma}
	
	\begin{proof}
		Let $X$ be the number of tuples $(u,v,w)$ in $\TT$ that satisfy $h_{u,v}(x)=w$. We have $X=g_x+\Bin(n-g_x,\pcol)+\Bin(n\lfloor\rho\rfloor,1/b)+\Bin(n,(\rho-\lfloor\rho\rfloor)/b)$, where
		\begin{enumerate}
			\item the first term is the $g_x$ real outputs from the users with input $x$,
			\item the second term is from the $n-g_x$ real outputs from the users with data not equal to $x$, each colliding with $x$ with probability $\pcol$, and
			\item the last two terms are from the blanket noises.
		\end{enumerate}
		Since $\hat{g}_x=(X-n\rho/b-n\pcol)/(1-\pcol)$, we have $\E[\hat{g}_x]=g_x$, i.e., $\PP^{\FE1}$ returns  unbiased estimates.
		
		Let $Y=\Bin(n-g_x,\pcol)+\Bin(n\lfloor\rho\rfloor,1/b)+\Bin(n,(\rho-\lfloor\rho\rfloor)/b)$ which is a sum of $n(\lfloor\rho\rfloor+2)-g_x$ independent Bernoulli variables. Define $\mu\coloneqq\E[Y]$, then  similar to the proof of Lemma~\ref{lem:accfe0}, we have
		\[
		\Pr \left[ |Y - \mu| > \max \left\{3\ln(2B/\beta), \sqrt{3\ln(2B/\beta) \cdot \mu} \right \} \right] < \frac{\beta}{B}.
		\]
		Therefore, 
		\begin{equation}\label{eq:fe1_chernoff}
			\begin{split}
				\Pr \left[ |\hat{g}_x - g_x| > \max \left\{3\ln(2B/\beta), \right. \right.\\
				\left. \left. \sqrt{3\ln(2B/\beta) \cdot \mu} \right \}/(1-\pcol) \right] < \frac{\beta}{B}.
			\end{split}
		\end{equation}
		Recall that $\pcol<1/b$, we have 
		\[\mu=\pcol(n-g_x)+n\rho/b\leq \frac{n}{b}+\frac{32\ln(2/\delta)}{\varepsilon^2}.\]
		Also note that $\pcol\leq1/2$, so inequality~\eqref{eq:fe1_chernoff} can be relaxed to 
		\[
		\Pr \left[ |\hat{g}_x - g_x| > \alpha \right] < \frac{\beta}{B}.
		\]
		Then the lemma follows by a union bound.
	\end{proof}
	
\begin{theorem}\label{the:fe1}
	For any $0 < \varepsilon \leq 3$, $0 < \delta < 1$, and $n, B, b\in \mathbb{N}$ such that $2\le b \le B/2$,  $\PP^{\FE1}$ is a private-coin $(\varepsilon, \delta)$-shuffle DP frequency estimation protocol that sends $1 + O\left(\frac{\log(1/\delta)}{\varepsilon^2} \cdot \frac{b}{n}\right)$ messages per user in expectation, each consisting of $O(\log B)$ bits. The frequency query of any element $x\in[B]$ can be answered in expected $O\left(n+\frac{\log(1/\delta)}{\varepsilon^2} \cdot b\right)$  time with error $O\left(\log{\frac{B}{\beta}} + \sqrt{\log {\frac{B}{\beta}} \left(\frac{n}{b} + \frac{\log(1/\delta)}{\varepsilon^2}\right)} \right)$.  The frequencies of all elements in $[B]$ can be found in expected $O\left(\left(\frac{n}{b} + \frac{\log(1/\delta)}{\varepsilon^2}\right) \cdot B\right)$ time.  
\end{theorem}
	
\begin{proof}
	It only remains to prove the last statement, which is faster than querying each element one by one by a factor of $b$. Such a faster all-frequency algorithm will also be useful for heavy hitter detection, which in turn will be needed for sparse vector summation.
	
	The idea is that, instead of computing the estimated frequency of each $x\in [B]$, we flip the problem around.  Recall that the variable $X$ in Algorithm \ref{algo:a-fe1} is equal to the number of $(u,v,w)$ tuples in $\TT$ such that $h_{u,v}(x)=w$. We initialize this variable to $0$ for all $x\in [B]$. Then, for each tuple $(u,v,w)\in\TT$, we find the set of all $x\in [B]$ such that $h_{u,v}(x)=w$ and increment their corresponding variables.  It can be verified that this set is
		\begin{equation*}
			\begin{split}
				\XX(u,v,w)\coloneqq\left\{u^{-1}(w+i\cdot b-v)\bmod{q}~\middle|~ \right. \\ 
				\left. i\in\left\{0,1,\dots,\left\lfloor{\frac{q-1-w}{b}}\right\rfloor\right\}\right\}\cap [B],
			\end{split}
		\end{equation*}
		where $u^{-1}$ is the multiplicative inverse of $u$ in $\mathbb{Z}_q$. 
		The set $\XX(u,v,w)$ has at most $q/b=O(B/b) $ elements and they can be found in $O(B/b)$ time. Meanwhile $\E[|\TT|] =n+n\rho=O\left(n+b\cdot\frac{\ln(1/\delta)}{\varepsilon^2}\right)$, so the total running time is $O\left(\left(\frac{n}{b} + \frac{\log(1/\delta)}{\varepsilon^2}\right) \cdot B\right)$. 
	\end{proof}
	
	

We compare our protocols with the state-of-the-art frequency estimation protocol \cite{badih2021on} in Table~\ref{tab:fecomp}.  To make the bounds more concise, we set a constant $\varepsilon$, and consider the parameter range $\Theta(\log(1/\delta))= \Theta(\log n) =\Theta(\log B) = \Theta(\log(1/\beta))$.
In particular, by taking $b=n/\log n$ in $\PP^{\FE1}$, our protocol achieves the same error as \cite{badih2021on}, but we reduce the number of messages to $O(1)$, reduce the message size by a $\log n$ factor, and reduce the query time of a single element by a $\log^3 n$ factor.  Recall that we require $b<B/2$, so for $B<2n/\log n$, we use the small-domain protocol, which obtains the same (or better) improvement.  By using a smaller $b$, the number of messages is further reduced to $1+o(1)$, although at the expense of a small increase in the error.  Furthermore, by imposing a dyadic decomposition on the domain, a frequency estimation protocol can be used for solving the \textit{range-counting} problem, so our result immediately improves the range-counting results in \cite{badih2021on} as well.  We omit the rather tedious details. 
	
	\subsubsection{Robust shuffle-DP} Since shuffle-DP only uses one-way communications, it naturally defends against a semi-honest adversary that may compromise any number of users. In the robust $(\varepsilon,\delta,\gamma)$ shuffle-DP model \cite{balcer2021connecting}, a fraction of $1-\gamma$ of the users are allowed to be malicious, i.e., they may not follow the specified protocol.  
	Our protocol can be easily made robust, since each users injects (expected) $\rho$ blanket noises independently.  By setting $\rho=\frac{32 \ln(2/\delta)}{\gamma\varepsilon^2} \cdot \frac{b}{n}$, we can still ensure that sufficient blanket noises are added.  Nevertheless, it is not possible to guarantee accuracy if users deviate from the protocol, as the malicious users may replace their data by fake values. The same situation happens for the MPC model with malicious parties, which can only guarantee privacy but not utility.  
	
\section{Heavy Hitter Detection}\label{sec:hhd}

A standard technique to reduce the heavy hitter detection problem to frequency estimation is to use a prefix tree \cite{cormode10frequent, erlingsson2019amplification} of $\log B$ levels over the domain $[B]$.  For any $x\in [B]$, we use $x[:i]$ to denote the first $i$ bits in the binary representation of $x$. For level $i$ of the prefix tree, we map every $x\in [B]$ to $x[:i] \in [2^i]$, and construct a frequency estimation data structure over the domain $[2^i]$.  It is clear that $x[:i]$ is a heavy hitter in level $i$ only if $x[:(i-1)]$ is a heavy hitter in level $i-1$.  To identify all the heavy hitters, we start from level $1$, which has only 2 elements in its domain, query for their frequencies, and then proceed to lower levels, only querying for the frequencies of elements whose prefixes are heavy hitters.  Adopting this idea directly for shuffle-DP, as done in \cite{badih2021on}, leads to an $O(\log^3 B)$-factor increase in the message number, because the privacy budget $\varepsilon$ has to be divided among the $\log B$ levels. 

We improve their heavy hitter detection algorithm using the following two ideas.  First, instead of  participating in all levels of the prefix tree, we ask each user to randomly pick one level.  This reduces the message number to $O(1)$.  A similar idea was used in the local-DP protocol \cite{bassily2017practically}, but our shuffle-DP protocol only uses private randomness where \cite{bassily2017practically} needs public randomness.  Next, we sample the messages to further reduce the message number to $o(1)$.  

Below we present our heavy hitter detection protocol $\PP^\HHD=(\RR^\HHD,\AAA^\HHD)$ for the case $B>n/\log n$. If $B\le n/\log n$, we can simply use the small-domain frequency estimation protocol without the need of heavy hitter detection.  We also assume $B$ and $n/\log n$ are both powers of $2$.
	
	
\subsection{Local randomizer}
Our heavy hitter detection randomizer $\RR^\HHD$  (Algorithm~\ref{algo:lr-hhd}) works as follows. Let $s=\log(n /\log n)$ be the first level, $t=\log B$ be the last level, and $r=t-s+1$ be the total number of levels. Suppose a user has element $x$. It first uniformly randomly chooses a level $i\in\{s,\dots,t\}$, and then inputs $x[:i]$ to the local randomizer of our frequency estimation protocol (with carefully chosen parameters, see Algorithm~\ref{algo:lr-hhd}). Then it adds the label $i$ to each message and sends it out with probability $p$, whose value will be determined by the analysis.
	
	\begin{algorithm} 
		\caption{Local Randomizer $\RR^\HHD$}\label{algo:lr-hhd}
		\KwPublic{$B, n, p, \varepsilon, \delta$}
		\KwIn{$x \in [B]$}
		\KwOut{A multiset $\TT \subseteq  \{s,\dots,t\}\times \{1,\dots,q-1\} \times [q] \times [b]$}
		$q\gets$ the smallest prime that is greater or equal to $B$\;
		$b\gets n/\log^2 n$\;
		$\TT \gets \varnothing$\;
		$i\gets$ Uniformly chosen in $\{s,\dots,t\}$\;
		\mbox{$\RR_i\gets\RR^{\FE1}$ with parameters $2^i,b,n/2r,\varepsilon, \delta/2$}\;
		$\TT_i \gets \RR_i(x[:i])$\; 
		\For{$(u_j, v_x, w_j) \in \TT_i$}{$y_j \gets \Ber(p)$ \; \If{$y_j = 1$}{$\TT \gets \TT \uplus \{(i, u_j, v_x, w_j)\}$\;}}
		\KwRet{$\TT$}\;
	\end{algorithm}
	
	\begin{lemma}[Privacy of $\RR^{\HHD}$]
		\label{lem:prihhd}
		$\RR^{\HHD}$ satisfies $(\varepsilon, \delta)$-shuffle DP if $n\ge 8r\ln\frac{2r}{\delta}$.
	\end{lemma}
	
	\begin{proof}
	    For any $i\in\{s,\dots,t\}$, let $n_i$ be the number of users that choose this level, then $n_i\sim\Bin(n,1/r)$. By Chernoff bound,
		$\Pr[n_i< n/2r]\le\exp(-n/8r)$. By union bound, $\Pr[\lor_{i=s}^t(n_i< n/2r)]\le r\exp(-n/8r)\le \delta/2$, i.e., the number of users on every level is at least $n/2r$ except with probability at most $\delta/2$. 
		Conditioned upon this happening, the messages at each level satisfies $(\varepsilon, \delta/2)$-shuffle DP by Theorem~\ref{thm:bibdp}.  The conditionality turns it into $(\varepsilon, \delta)$-DP.  Finally, since no user participates in more than one level, the privacy of all levels follows from parallel composition of DP.
	\end{proof}
	
	\subsection{Analyzer}
	Now we introduce the analyzer $\AAA^{\HHD}$ (Algorithm \ref{algo:a-hhd}). First it classifies the messages to the correct levels according to their labels. Then starting from $C_s=\{0,1\}^s$, it recovers the heavy hitter candidates bit by bit.
	For each $i=s,\dots,t$, it counts the occurrences of all candidates using Algorithm~\ref{algo:c-hhd}, which is the same as the analyzer of our frequency estimation protocol but without bias removal.  All candidates that appear less than $\Delta$ times are removed, and the remaining candidates are extended by one bit to obtain $C_{i+1}$ except in the last level. The final set of candidates  is denoted as $H$.
	
	\begin{algorithm}
		\caption{Counter $\CC^{\HHD}$}\label{algo:c-hhd}
		\KwPublic{$B, b$}
		\KwIn{An element $x$; A multiset $\TT\subseteq\{1,\dots,q-1\}\times[q]\times[b]$}
		\KwOut{Occurrences $X$ of $x$}
		$q\gets$ the smallest prime that is greater or equal to $B$\;
		$X \gets 0$\;
		\For{$(u,v,w) \in \TT$}{ \If{$h_{u,v}(x) = w$}{$X \gets X + 1$\;}}
		\KwRet{$X$}\;
	\end{algorithm}
	
	\begin{algorithm} 
		\caption{Analyzer $\AAA^\HHD$}\label{algo:a-hhd}
		\SetKwComment{Comment}{$\triangleright$\ }{}
		\KwPublic{$B, n, p$}
		\KwIn{A multiset $\TT \subseteq \{s,\dots,t\} \times \{1,\dots,q-1\} \times [q] \times [b]$}
		\KwOut{The set of heavy hitters candidates  $H \subseteq [B]$}
		$q\gets$ the smallest prime that is greater or equal to $B$\;
		$\Delta \gets p \phi n/2r$\;
		\For{$i \gets s$ \KwTo $t$}{$\TT_i \gets \{ (u, v, w)\ |\ (i, u, v, w) \in \TT\}$\;}
		$C_s \gets\{0,1\}^s$\;
		\For{$i \gets$ $s$ \KwTo $t-1$}{
		Set the parameters of $\CC^\HHD$ to be $2^i,n/\log^2 n$\;
		$C_i' \gets \{x\in C_i\mid \CC^\HHD(x,\TT_i)\ge\Delta\}$\;
		$C_{i+1} \gets \{\overline{x0}, \overline{x1} \mid x \in C_{i}'\}$\;
		}
		Set the parameters of $\CC^\HHD$ to be $2^t,n/\log^2 n$\;
		$H\gets\{x\in C_t\mid \CC^\HHD(x,\TT_t)\ge\Delta\}$\;
		\KwRet{$H$}\;
	\end{algorithm}
	
	\begin{lemma}[Completeness of $\PP^\HHD$]
		\label{lem:hhdacc}
		For any $\beta\in(0,1)$, if $p \ge \frac{8r}{\phi n}\ln\frac{r}{\phi\beta}$, then $H$ contains all the heavy hitters with probability at least $1-\beta$.
	\end{lemma}
	\begin{proof}
	We use $g(x[:i])$ to denote the frequency of $x[:i]$ in the $i$-th level.  Then for any heavy hitter $x$, $g(x[:i])\ge \phi n$.  Let $X_i$ be the counter for $x$ returned by $\CC^\HHD$.  We bound the probability $\Pr[X_i<\Delta]$.  Since blanket noises and hash collisions can only make $X_i$ larger, it suffices to bound $\Pr[\Bin(\phi n, p/r) < \Delta]$. By Chernoff bound, this is at most  $\exp(-p\phi n/8r) \le \phi\beta/r $.
    Then the probability that $x$ is removed from the set of candidates follows from a union bound:
	\[
		\Pr[x\notin H]=\Pr \left[\bigvee_{i=s}^{t} (X_i < \Delta) \right]\le \sum_{i=s}^{t}\Pr [X_i < \Delta] \le \phi\beta.
	\]
	Finally, applying another union bound over all the at most $1/\phi$ heavy hitters yields the lemma.
	\end{proof}
	
	\begin{lemma}[Efficiency of $\PP^\HHD$] \label{lem:hhdeff}
	Assume $\log n =\Theta(\log B)= \Theta(\log(1/\delta))$ and $\varepsilon$ is a constant. 
	If $\phi=\Omega(\log^2 n/n)$ and $p \ge \frac{30r}{\phi n}\ln (\phi B)$, then
	\begin{enumerate}
		\item $\E[|C_i|]=O(1/\phi)$ for each $i=s+1,\dots,t$;
		\item the expected number of messages per user is $O(p)$;
		\item the expected running time of $\AAA^\HHD$ is $O(pn/\phi)$.
	\end{enumerate}
	\end{lemma}
	
	\begin{proof}
		First we prove (1). We bound $Y$, the number of elements whose occurrences exceed  $\Delta$. Let $n_i$ be the number of users that have chosen level $i$. Define $n_0\coloneqq 2n/r$. By Chernoff bound, $\Pr[n_i>n_0]\le\exp(-n/3r)=O(1/Bn)$.  Next we assume $n_i\le n_0$ for all $i$. Note that more users chosen in level $i$ would increase $Y$, so we only need to consider the case $n_i=n_0$. We fix any combination of $n_0$ users and prove the result. For any $x\in\{0,1\}^i$, denote $g_i(x)$ as the true frequency of $x$ in the $n_0$ users and $X$ as the occurrences estimated by Algorithm~\ref{algo:c-hhd}. Then
		\begin{align*}
		    X&\sim\Bin(g_i(x),p)+\Bin\left(n_0-g_i(x),p\pcol\right)\\
		    &+\Bin\left(n_0\lfloor \rho \rfloor,p/b\right)+\Bin\left(n_0,p(\rho-\lfloor \rho \rfloor)/b\right),
		\end{align*}
		where $b=n/\log^2 n$ and
		\[\rho=\frac{32 \ln(4/\delta)}{\varepsilon^2} \cdot \frac{b}{n/2r}=O(1).\]
		Recall that $\pcol<1/b$ is the probability of hash collision. 
		Therefore, $\E[X]\le p(g_i(x)+n_0\pcol+n_0\rho/b)= p g_i(x)+O(pn_0/b).$ 
		Assume $g_i(x)<\phi n_0/8$. Since $\phi=\Omega(\log^2 n/n)=\Omega(1/b)$, there exists large enough $n$ such that $O(pn_0/b)<p\phi n_0/24$. Therefore, $\E[X]\le p\phi n_0/8+p\phi n_0/24=p\phi n_0/6=2\Delta/3$. Let $c=\Delta/\E[X]-1\ge 1/2$. By Chernoff bound,
		\begin{align*}
			\Pr[X \geq \Delta]&=\Pr[X \geq (1+c)\E[X]]\le \exp \left(-\E[X] \cdot\frac{c^2}{2+c}\right)\\
			&=\exp \left(-\frac{\Delta c^2}{(2+c)(1+c)}\right)\le \exp\left( -\frac{\Delta}{15}\right) \le \frac{1}{\phi B}.
		\end{align*}
		Note that the number of elements $x$ with $g_i(x)\ge \phi n_0/8$ is at most $8/\phi$, and $|C_i|\le B$. Therefore, we have 
		\begin{align*}
		    \E[|C_i|]&\le\Pr[n_i>n_0]\cdot \E[|C_i|\mid n_i>n_0]+\E[|C_i|\mid n_i\le n_0]\\
		    &\le O(1/ Bn)\cdot B + (1/\phi B)\cdot B+8/\phi= O(1/\phi)
		\end{align*}
		
		Next we prove (2). 
		Each user chooses exactly one level, and in this level, each user is expected to send $p(1+\rho)$ messages where $\rho=O(1)$, so the total number of messages each user sends is
		$O(p).$
		
		Finally we prove (3). The analyzer runs in three steps:
		\begin{enumerate}
		    \item Classify the messages to the corresponding levels. The cost is linear to the number of total messages, which is $O(pn)$.
		    \item Estimate the occurrences of all elements in $[2^s]$. By Theorem~\ref{the:fe1}, the cost is \[O(p\cdot 2^s(\E[n_s]/b+\log(1/\delta)/\varepsilon^2))=O(p n).\]
		    \item Filter the set of candidates and extend it to the next level. For the $i$-th level, note that 
    		\begin{align*}
    		    \E[|C_i|\cdot n_i]&\le\Pr[n_i>n_0]\cdot Bn+n_0\cdot \E[|C_i|\mid n_i\le n_0]\\
    		    &\le O(n_0/\phi).
    		\end{align*}
		    Therefore, the expected cost is
		    \[O\left(\sum_{i=s}^t p\E[|C_i|\cdot n_i]\right)=O(pn/\phi).\]
		\end{enumerate} 
        Summing up all costs, the expected total running time is $O(pn/\phi)$.
	\end{proof}
	
	\begin{theorem}\label{thm:hhd}
		Assume $\log n =\Theta(\log B)= \Theta(\log(1/\delta)) = \Theta(\log(1/\beta))$, $\phi=\Omega(\log^2 n/n)$, and $\varepsilon$ is a constant. $\PP^\HHD$ a private-coin $(\varepsilon, \delta)$-shuffle DP protocol, in which each user sends $O(\log^2 n/\phi n)$ messages in expectation, each consisting of $O(\log B)$ bits. All the heavy hitters are reported in $O(\log^2 n/\phi^2)$ time with probability at least $1-\beta$.  
	\end{theorem}
	
	\begin{proof}
	We set 
	\[p=\frac{30r}{\phi n}\ln\left(\frac{r}{\phi\beta}+ \phi B\right)=\Theta\left(\frac{\log^2 n}{\phi n}\right),\]
	which is less than $1$ for  $\phi=\Omega(\log^2 n/n)$. Then $p$ satisfies both the requirements of Lemma~\ref{lem:hhdacc} and \ref{lem:hhdeff}, and we conclude this theorem by Lemma~\ref{lem:prihhd}--\ref{lem:hhdeff}.
	\end{proof}
	
By also running our frequency estimation protocol (using half of the privacy budget, say), we can estimate the frequencies of all elements with $1+o(1)$ messages per user in total time $O(n^2)$, for any $B=n^{O(1)}$:
	\begin{corollary}
		\label{col:hhdforfe}
        For any $B=n^{O(1)}$ and any constant $c>2$, there is a $1+o(1)$-message $(\varepsilon,\delta)$-shuffle DP protocol that estimate the frequency of all elements with $O(\log^c n)$ error in time $O(n^2/\log^c n)$.
	\end{corollary}
	
	\begin{proof}
	    Take $\phi=\log^c n/n$ in $\PP^\HHD$ and $b=n/\log^3 n$ in $\PP^{\FE1}$.
	        Then the number of messages per user is $1+O(1/\log^2 n)+O(\log^2n/\phi n)=1+o(1)$.  The frequencies of all the heavy hitters are given by $\AAA^{\FE1}$ with error $O(\log^2 n)$.
	        The frequencies of all the light hitters are estimated as 0, so the error is $O(\phi n)=O(\log^c n)$. The running time of $\AAA^{\HHD}$ and $\AAA^{\FE1}$ are $O(\log^2 n/\phi^2)$ and  $O(n/\phi)$ respectively, so the total query time is $O(n^2/\log^c n)$.
	\end{proof}
	
\section{1-Sparse Vector Summation}

\subsection{Review of \texorpdfstring{\citet{ghazi21aggregation}}{Ghazi et al.}}
\label{sec:review}
Recall that in the real summation problem, the $i$-th user holds a real number $w_i\in [0, 1]$, and the goal is to estimate their sum.  Let $\Delta\ge 1$ be a quantization parameter. The shuffle-DP protocol of \cite{ghazi21aggregation} sends the following messages:
\begin{enumerate}
    \item sends $\lfloor \Delta w_i \rfloor+\Ber(\Delta w_i-\lfloor \Delta w_i \rfloor)$ if this value is nonzero;
    \item sends the central-DP noise, which is nonzero with probability $\softO(\Delta/n)$ (they show that the sum of these noises almost matches the discrete Laplace distribution); and
    \item sends a random multiset $S$ of zero-sum noises to ensure shuffle-DP, where each $s\in S$ is taken from  $\{-\Delta,\dots,\Delta\} - \{0\}$ and $\E[|S|]=\softO(\Delta/n)$.
\end{enumerate}
The analyzer simply adds up all messages received and scales it down by $\Delta$. The estimator is unbiased with variance $O(1+\sum_i w_i/\Delta^2)$, where $O(\sum_i w_i/\Delta^2)$ comes from (1) and $O(1)$ comes from (2).

For the 1-sparse vector summation problem, the $i$-th user holds an element $x_i\in[B]$ together with a weight $w_i \in [0, 1]$, and the goal is to obtain an estimate $\hat{g}_x$ for $g_x = \sum_{i=1}^n w_i\cdot \mathbb{I}[x_i = x]$ for every $x \in [B]$ so as to minimize the MSE $=\sum_x (\hat{g}_x - g_x)^2$.
\citet{ghazi21aggregation} simply apply their real summation protocol on each coordinate $x\in[B]$, each with $(\varepsilon/2, \delta/2)$ privacy budget, resulting in a message number of $1+\softO(\Delta B/ n)$ with MSE $O(B+\sum_i w_i/\Delta^2)=O(B+n/\Delta^2)$. 

\citet{ghazi21aggregation} set $\Delta=\sqrt{n}$ so the variance on each coordinate is $O(1)$.  However, we observe that setting $\Delta=\sqrt{n/B}$ does not increase the MSE asymptotically (it is still $O(B)$), while the number of messages can be reduced to $1+\softO(\sqrt{ B/ n})$, which is $1+o(1)$ when $B=o(n)$. Thus, they can achieve the optimal central-DP MSE with $1+o(1)$ messages for the low-dimensional setting $B=o(n)$.

In the high-dimensional setting $B=\Omega(n)$, the optimal central-DP MSE is $\softO(n)$, by simply running the Laplace mechanism followed by setting all estimates less than $\log B / \varepsilon$ to $0$.  We note that this ``zeroing'' post-processing step can also be applied to \cite{ghazi21aggregation}, so that their MSE can be reduced to $\softO(n)$ when $B>n$.  However, their message number is at least $\softOmega( B/ n)$ even if we set $\Delta=1$.

\subsection{Our protocol}
\label{sec:vector}
Below, we show how to achieve $\softO(n)$ MSE with $1+o(1)$ messages in the high-dimensional setting, based on our frequency estimation and heavy hitter detection protocols.

User $i$ with input $(x_i,w_i)$ applies randomized rounding to $w_i$. More precisely, we extend the domain to $[B+1]$ with a dummy element $B$, and then convert $x_i$ to $f(x_i, w_i)$ where
\[
    f(x, w) = \begin{cases}
    x, & \text{with probability }w; \\
    B, & \text{otherwise}.
    \end{cases}
\]
Then we just apply the protocol from Corollary~\ref{col:hhdforfe}, namely, heavy hitter detection $+$ frequency estimation.  

\begin{theorem} \label{thm:vector_sum}
When $B=\Omega(n)$, there is a $(\varepsilon, \delta)$-shuffle DP protocol for the 1-sparse vector summation problem that sends $1+o(1)$ messages while achieving MSE $O(n\log^c n)$ for any $c>2$. The analyzer's running time is $O(n^2/\log^c n)$.
\end{theorem}

\begin{proof}
    Since for any two neighboring inputs, after applying $f$, the outputs are also neighboring if not identical, so the privacy of the protocol still holds. The message number and running time directly follow Corollary~\ref{col:hhdforfe}.  Below we analyze the MSE. 
    
    Let $H$ be the set of heavy hitter candidates, then given $x\in H$, the frequency of $x$ is estimated by our frequency estimation protocol. By the proof of Lemma~\ref{lem:accfe1}, our estimator is \[\hat{g}_x=(X-n\rho/b-n\pcol)/(1-\pcol),\]
    where 
    \[X=\left(\sum_{i=1}^n X_i\right)+\Bin(n\lfloor\rho\rfloor,1/b)+\Bin(n,(\rho-\lfloor\rho\rfloor)/b),\]
    $b=n/\log^2 n$,
    $X_i=z_{xi}+(1-z_{xi})\Ber(\pcol)$,
    and $z_{xi}=\mathbb{I}[f(x_i,w_i)=x]$. 
    Since $X_i^2=X_i$, we have $\Var[X_i]=\E[X_i]-\E[X_i^2]\le \E[X_i]$, while binomial distributions have the similar inequality. 
    Since $\E[\hat{g}_x]=g_x$, and $X_i,X_j$ are independent for all $i\neq j$, we have 
    \begin{align*}
        \E[(\hat{g}_x-g_x)^2\mid x\in H]&=\Var(\hat{g}_x\mid x\in H)\\
        &\le\E[X]/(1-\pcol)^2\\
        &\le (g_x+(n-g_x)\pcol+n\rho/b)/(1-\pcol)^2\\
        &=O(g_x+\log^2 n).
    \end{align*}
    
    Let $\phi$ be the threshold chosen in Corollary~\ref{col:hhdforfe}, then $\phi n=\log^c n=\omega(\log^2 n)$.
    We set $\beta=1/n$ in Lemma~\ref{lem:hhdacc}, then with probability at least $1-1/n$, all $x$ with $\bar{g}_x\ge\phi n$ are in the candidate set, where $\bar{g}_x=\sum_{i=1}^n z_{xi}$. On the other hand, by Chernoff bound, $\Pr[\bar{g}_x<g_x/2]=\exp(-\Omega(g_x))$. When $g_x\ge 2\phi n=\omega(\log^2 n)$, it is $\Pr[\bar{g}_x<\phi n]=\exp(-\omega(\log^2 n))=O(1/n)$, which implies $\Pr[x\notin H]\le \Pr[\bar{g}_x<\phi n]+\Pr[x\notin H\mid \bar{g}_x\ge \phi n]=O(1/n)$. For any $x\notin H$, our estimator is simply $\hat{g}_x=0$. Therefore,
    \begin{align*}
        \MSE&=\sum_{x\in[B]} \E[(\hat{g}_x-g_x)^2]\\
        &\le\sum_{x:g_x\ge2\phi n}(\Var(\hat{g}_x\mid x\in H]+g_x^2\cdot\Pr[x\notin H])+\sum_{x:g_x<2\phi n}g_x^2\\
        &=O\left(\sum_{x\in[B]}g(x)+\frac{\log^2 n}{\phi}+\frac{1}{n}\sum_{x\in[B]}g_x^2+\sum_{x:g_x\le 2\phi n}g_x^2\right)\\
        &=O\left(n+n+n+2\phi n\sum_{x\in[B]}g_x\right)\\
        &=O(n\log^c n).\qedhere
    \end{align*}
\end{proof}

\paragraph{Remark 1}\label{remark1} Instead of using heavy hitter detection $+$ frequency estimation, we could also use the frequency estimation protocol alone (followed by the zeroing post-processing step), which can improve the MSE by a polylogarithmic factor but at the expense of a running time linear in $B$. 

\paragraph{Remark 2} Our protocol can be easily extended to the case where each user has a $k$-sparse vector.  For \textit{coordinate-level} DP (i.e., neighboring datasets differ by one coordinate in some user's vector), we can treat each of the $k$ coordinates as a 1-sparse vector, which increases the message number by a factor of $k$; for \textit{vector-level} DP (i.e., neighboring datasets may differ by the whole vector owned by some user), we can use group privacy \cite{vadhansurvey} to scale $(\varepsilon, \delta)$ down to $(\varepsilon/k, \delta/ke^{\varepsilon})$ to reduce this case to coordinate-level DP.

\section{Experiments}

We have implemented our protocols and the protocols in \cite{badih2021on,ghazi21aggregation} in C++\footnote{\url{https://github.com/hkustDB/SDPFE}}, and  conducted all the experiments on a server with an Intel Xeon Silver 4116 CPU.  Note that, however, the accuracy and message number of the protocols do not depend on the machine, only the running time of the analyzer does.  The randomizer's time is negligible. 

We used both real and synthetic data in the experiments:

\begin{itemize}
    \item The AOL dataset \cite{pass2006picture} is a collection of real-world website accesses.  We truncate each URL to a prefix of a certain length, so as to obtain different domain sizes.  For example, truncating to the first 3 characters generates a domain of size $B=2^{24}$, since each character has 8 bits. 
    \item We used the Zipfian distribution to generate synthetic data of varying levels of skewness.  We first randomly permute the elements in $[B]$.  Then we draw elements following the distribution Zipf($\alpha, B$), whose probability mass for the $k$-th element in the permutation is $k^{-\alpha} / \sum_{i=1}^B i^{-\alpha}$, so a larger $\alpha$ corresponds to more skew.    
\end{itemize}

\subsection{Frequency Estimation}

\begin{table*}
	\centering
	\begin{tabular}{|c||c|c|c|c|c|c||c|c|c|c|c|c|}
		\hline
		& \multicolumn{6}{|c||}{95\% error} & \multicolumn{6}{|c|}{\#Messages / user}  \\
		\hline
		$\varepsilon$ & 0.25 & 0.5 & 0.75 & 1 & 2 & 3 & 0.25 & 0.5 & 0.75 & 1 & 2 & 3  \\
		\hline \hline
		GGKPV & - & 98.68 & 65.67 & 49.24 & - & - & - & 3317 & 1475 & 830 & - & - \\
		Ours ($c=1$) & 72.99 & 38.94 & 25.83 & 22.47 & 13.68 & 12.88 & 114.82 & 30.65 & 15.010& 9.51  &4.23& 3.29 \\
		Ours ($c=3$) &1395.4 & 411.21 & 228.28 & 163.94 & 102.20 & 92.24 & 1.85 & 1.22 & 1.10& 1.063& 1.024& 1.017 \\
		\hline \hline
		& \multicolumn{6}{|c||}{Query time (all / s)} & \multicolumn{6}{|c|}{Query time (single /s)}  \\
		\hline
		$\varepsilon$ & 0.25 & 0.5 & 0.75 & 1 & 2 & 3 & 0.25 & 0.5 & 0.75 & 1 & 2 & 3  \\
		\hline \hline
		GGKPV & $> 1$ day & 99531.8 & 42581.2 & 23908.6 & - & - & - & 48.44 & 16.45 & 13.11 & - & -\\
		Ours ($c=1$) & 956.32 & 231.51 & 118.27 & 70.91 & 31.69 & 23.91 & 0.359 &  0.095&  0.047 &  0.028 &  0.013  & 0.009\\
		Ours ($c=3$) &1106.27 & 1140.04 & 1182.6 & 1047.19 & 1040.12 & 1151.4 & 0.003 &  0.003 &  0.003&  0.003&  0.003 &  0.003 \\
		\hline
	\end{tabular}
	\caption{Comparison among frequency estimation protocols varying privacy parameter $\varepsilon$ ($n=10^5, B=2^{24}, \delta=10^{-10}, b=10^5/16.6^c$)}
	\label{tab:diff_epsilon}
\end{table*}

\begin{table*}
    \centering
    \begin{tabular}{|c||c|c|c|c||c|c|c|c|}
        \hline
        & \multicolumn{4}{|c||}{95\% error} & \multicolumn{4}{|c|}{\#Messages / user}  \\
        \hline
        $n$ & $10^4$ & $10^5$ & $10^6$ & $10^7$ & $10^4$ & $10^5$ & $10^6$ & $10^7$ \\
        \hline \hline
        GGKPV & 39.12 &  49.24 &  60.14 &  - &  665 &  830 &  996 &  - \\
        Ours ($c=1$) & 18.40 & 22.47 & 27.32 & 29.64 & 9.265 & 9.51 & 9.65 & 9.75 \\
        Ours ($c=3$) & 132.55 & 163.94 & 204.57 & 247.23 & 1.09 & 1.064 & 1.045 & 1.034 \\
        \hline \hline
        & \multicolumn{4}{|c||}{Query time (all / s)} & \multicolumn{4}{|c|}{Query time (single / s)}  \\
        \hline
        $n$ & $10^4$ & $10^5$ & $10^6$ & $10^7$ & $10^4$ & $10^5$ & $10^6$ & $10^7$ \\
        \hline \hline
        GGKPV & 11702.9 & 23908.6 & 53752.6 & $> 1$ day & 0.55 & 16.43 & 225.41 & - \\
        Ours ($c=1$) & 57.1 & 70.91 & 98.22 & 139.47 & 0.0028 &  0.0282 &  0.2836 &  2.9131 \\
        Ours ($c=3$) & 545.43 & 1047.19 & 1839.29 & 2859.06 & 0.0003 &  0.0031 &  0.0312 &  0.3119 \\
        \hline
    \end{tabular}
    \caption{Comparison among frequency estimation protocols varying number of users $n$ ($B=2^{24}, \varepsilon = 1, \delta = n^{-2}, b=n/\log^c n$)}
    \label{tab:diff_n}
\end{table*}

We evaluate our large-domain protocol against the start-of-the-art method \cite{badih2021on}, denoted as GGKPV. We compare them along the following metrics: (1) error, including the maximum error, the 95\% error, the 90\% error, and the median error over all elements in $[B]$; (2) the expected number of messages sent per user; (3) query time to estimate the frequency of a single element; and (4) the query time to estimate all frequencies.

The results on the AOL dataset are shown in Figure \ref{fig:msg_error}.  Note that the error is in linear scale while the message number is shown in log scale.  Recall that the parameter $c$ controls the error-message number trade-off of our protocol.  In particular, $c=1$ has $O(\log n)$ error (same as that of GGKPV) with $O(1)$ messages (GGKPV has $O(\log n)$ message).  The experimental results suggest that our actual error is 
smaller than half that of GGKPV, suggesting that the hidden constant in the $O(\log n)$ error in our protocol is smaller.  The message number is $1/100$ that of GGKPV, which is also more than what the theory predicts.  Part of the reason is the parameter optimization (see the remark after Theorem \ref{thm:bibdp}), which reduces the message number by $70\%$--$90\%$. Choosing a larger $c$ tunes the trade-off: With $c=3$, each user sends $1+o(1)$ messages, albeit at the expense of an $O(\log^2 n)$ error.  

\begin{figure}
	\begin{center}
		\includegraphics[width=0.92\linewidth]{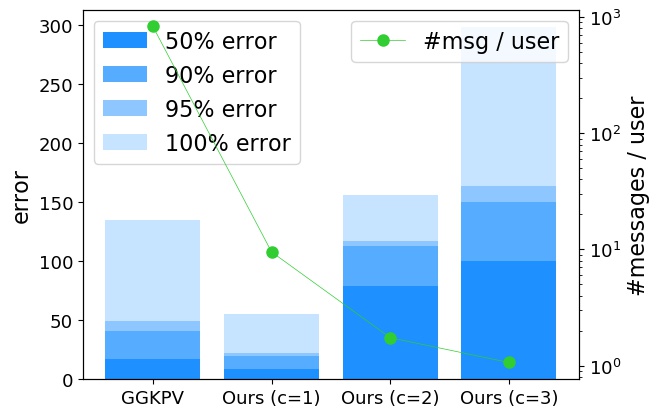}
		\caption{Comparison among frequency estimation protocols on error and \#messages/user ($n=10^5, B=2^{24}, \varepsilon=1, \delta=n^{-2}$)}
		\label{fig:msg_error}
	\end{center}
    \vspace{-0.5em}
\end{figure}

Table \ref{tab:diff_epsilon} shows the results on different $\varepsilon$. GGKPV only supports $\varepsilon \leq 1$, while for very small $\varepsilon$, it cannot finish within 1 day.
From the results, we see that, a larger $\varepsilon$ leads to smaller error, less messages, and faster query time, which are as expected.  The more than 100x reduction of query time is also very significant and useful in practice.
Table \ref{tab:diff_n} shows results for different $n$, the number of users. As the theory predicts, message number, the error, and the query time of all elements of our protocol are logarithmically affected by $n$, while the query time of a single element is linear to $n$.

\subsection{Heavy Hitter Detection}

\begin{table*}
    \centering
    \begin{tabular}{|c||c|c|c|c|c|c|}
        \hline
        Protocols & \multicolumn{4}{|c|}{Our HHD} & GGKPV-HHD & Our FE (c=3) \\
        \hline
        Threshold $\phi$ & 0.01 & 0.0075 & 0.005 & 0.0025 & - & - \\
        \hline \hline
        \#Messages / user & 0.21 & 0.28 & 0.40 & 0.77 & $> 10^8$ & 1.034 \\
        Message size / user (bytes) & 3.57 & 4.95 & 7.08 & 13.85 & $>20$ GB & 18.61 \\
        Query time (s) & 41.29 & 190.45 & 612.85 & 4802.28 & $> 1$ year & $> 1$ year\\ 
        Recall rate  & 100\% & 100\% & 100\% & 100\% & 100\% & 100\% \\
        \hline
    \end{tabular}
    \caption{Comparison among heavy hitter detection protocols ($n=10^7, B=2^{48}, \varepsilon=1, \delta = 10^{-14}$)}
    \label{tab:comparehhd}
\end{table*}

Table \ref{tab:comparehhd} compares different protocols for heavy hitter detection.  We use a large domain $B=2^{48}$. For a smaller $B$, we could just use our frequency estimation protocol, but for such a large $B$, it would take more than 1 year (estimated) for it to obtain all frequencies.  Although the message number of GGKPV-HHD \cite{badih2021on} is polylogarithmic, it is gigantic number in practice.  On the other hand, our HHD protocol provides a highly-efficient and light-weight solution to this problem.  The recall rate is 100\% in all our experiments, namely, all heavy hitters have been identified.  There are false positives, though, but they can be removed by querying the frequency estimation data structure, subject to the error in the estimated frequencies.  The costs of our HHD also depend on $\phi$, where a larger $\phi$ reduces all costs, as predicted by the theoretical analysis.

\subsection{Sparse Vector Summation}
\begin{table*}
	\centering
	\begin{tabular}{|c|c||c|c|c|}
		\hline
		\multicolumn{2}{|c||}{Mechanisms} & \# Messages / user & MSE ($ \times 10^8$ ) & Query time (s) \\
		 \hline
		\multirow{2}{*}{GKMPS} & DLap & \multirow{2}{*}{16788.66} & 134.11 & \multirow{2}{*}{1600}\\ 
		 & Zeroing & & 0.65 & \\		\hline
		\multicolumn{2}{|c||}{HHD+FE} & 1.15 & 4.30 & 548\\
		\hline
		\multirow{3}{*}{FE} & $c=1$ & 9.75 & 0.09 & 101\\ 
		 & $c=2$ & 1.54 & 0.67 & 109 \\ 
		 & $c=3$ & 1.03 & 1.54 & 1349 \\ 
		\hline
	\end{tabular}
	\caption{Comparison among  vector summation protocols ($n=10^7, B=2^{24}, \varepsilon=1, \delta = 10^{-14}, \phi=7\times 10^{-4}$)}
	\label{tab:vscomp}
\end{table*}

\begin{table*}
	\centering
	\begin{tabular}{|c||c|c||c|c|c||c|c|}
		\hline
		Metrics & \multicolumn{2}{|c||}{\# Messages / user} & \multicolumn{3}{|c||}{MSE ($ \times 10^8$ )} & \multicolumn{2}{|c|}{Query time (s)}\\
		\hline
		\multirow{2}{*}{Mechanisms} & \multirow{2}{*}{HHD+FE} &\multirow{2}{*}{GKMPS} & \multirow{2}{*}{HHD+FE} &  \multicolumn{2}{|c||}{GKMPS} & \multirow{2}{*}{HHD+FE} &\multirow{2}{*}{GKMPS}\\
		\cline{5-6}
		& & & & DLap & Zeroing & & \\
		\hline \hline
		AOL & \multirow{4}{*}{1.14} & \multirow{4}{*}{$>10^6$} &  18.01 & \multirow{4}{*}{8590}  &  3.69 & \multirow{4}{*}{1724} & \multirow{4}{*}{> 1 day}\\
		Zipf ($\alpha = 1$) &  & & 5.50 & & 1.07 & &\\
		Zipf ($\alpha = 2$) & & & 0.42 & &  0.06 & &\\
		Zipf ($\alpha = 3$) & &  & 0.08 & &  0.01 & &\\
		\hline
	\end{tabular}
	\caption{Comparison among vector summation protocols varying datasets ($n=10^7, B=2^{30}, \varepsilon=1, \delta = 10^{-14}, \phi=7\times 10^{-4}$)}
	\label{tab:ldcomp}
\end{table*}
To conduct experiments on the 1-sparse vector summation problem, we added a weight to each user's element.  For the AOL dataset, each user reports its most visited URL and the weight is the number of times visited normalized to $(0,1)$. For Zipfian data, we added uniform random numbers in $(0, 1)$ as weights.

As described in Section \ref{sec:vector}, we can solve the 1-sparse vector summation problem by combining heavy hitter detection and frequency estimation (denoted as HHD+FE) for very large $B$, or by frequency estimation alone for $n<B<\softO(n^2)$ (denoted FE; see \nameref{remark1}).  We compare them with the protocol in \cite{ghazi21aggregation} (denoted as GKMPS), which simply executes $B$ instances of real summation with $\Delta = 1$.  Each real summation generates a discrete Laplace noise with scale $2/(1-\gamma)\varepsilon$, where $\gamma$ is a parameter that controls the trade-off between error and message number. We set $\gamma = 0.9$ following \cite{ghazi21aggregation}.  Directly doing so results in $\softO({B})$ MSE (denoted as ``DLap''), while applying zeroing post-processing step as described in Section \ref{sec:review},
which reduces the MSE to $\softO(n)$. Our experiments demonstrate that this zeroing post-processing step significantly reduces their MSE in practice.  We do the same zeroing post-processing step for HHD+FE and FE as well.

We first compare all the protocols on the AOL dataset with $B$ slightly larger than $n$, and
the results are shown in Table~\ref{tab:vscomp}.
Although theoretically GKMPS sends $1+\softO(B/n)$ messages, this number in the experiments is more than $10,000$, indicating an impractical hidden polylogarithmic factor. At the same time, our protocols send much fewer messages while achieving a similar or smaller error. 
By taking $c=3$, both HHD+FE and FE send $1+o(1)$ messages while achieving a near-optimal MSE.  Setting a smaller $c$ reduces the MSE to be smaller than that of GKMPS while the message number increases to a small constant.  Comparing our two protocols,  FE outperforms HHD+FE, because the later assigns half of the privacy budget to detect the heavy coordinates. 

Next, we tested with a $B\gg n$ on the AOL dataset and the Zipfian datasets and the results are shown in Table~\ref{tab:ldcomp}.  For such a large $B$, neither FE nor GKMPS can finish in a day.  
The MSE of GKMPS is computed by assuming no error over all the $0$ coordinates, which yields a (small) underestimate of their actual MSE.
The results on the Zipfian datasets show skewness affects the performance.  We observe that the number of messages and the query time of HHD+FE do not really change, but the MSE of both HHD+FE and GKMPS (with zeroing) decreases on more skewed data. This is because as skewness increases, the MSE from those coordinates less than the zeroing threshold also reduces.



\section{Open Problems}

Results in this paper have further improved the power of shuffle-DP with just $1+o(1)$ messages per user.  A major open problem is thus: are there any natural problems for which the optimal central-DP error is not achievable (subject to polylogarithmic factors) with $1+o(1)$ messages?  

There are also some interesting but more technical questions that are still open: (1) For frequency estimation, our $1+o(1)$ protocol still does not match the optimal central-DP error of $O(\log n)$, but with an $\omega(1)$ gap.  Closing this gap could be an interesting theoretical problem.  (2) For $k$-sparse vector summation under vector-level DP, the use of group privacy perhaps does not yield the best dependency on $k$.  We note that this problem is equivalent to (the sparse case of) mean estimation in high dimensions, and techniques from \cite{huang2021instance,biswas2020coinpress, zhou2022locally} might be useful for further improvement.

\begin{acks}
This work has been supported by HKRGC under grants 16201819, 16205420, and 16205422. We thank the anonymous reviewers of CCS '22 for valuable suggestions on improving the presentation of the paper.
\end{acks}

\bibliographystyle{ACM-Reference-Format}
\bibliography{ref}

\end{document}

%% file: macro.tex
\usepackage{amsmath,amsfonts}
\usepackage{algorithmic}
\usepackage{graphicx}
\usepackage{textcomp}
\usepackage[utf8]{inputenc} 
\usepackage[T1]{fontenc}  
\usepackage{hyperref}     
\usepackage{amsfonts}       
\usepackage{graphicx}
\usepackage{bm}
\usepackage[ruled,linesnumbered]{algorithm2e}
\usepackage{multirow}
\usepackage{amsthm}
\usepackage{subfigure}

\newcommand{\E}{\mathbb{E}} 
\newcommand{\BIB}{\mathrm{BIB}}
\newcommand{\FE}{\mathrm{FE}}
\newcommand{\Var}{\mathrm{Var}}

\newcommand{\Ber}{\mathrm{Ber}}
\newcommand{\Bin}{\mathrm{Bin}}

\newcommand{\HHD}{\mathrm{HHD}}

\newcommand{\Noise}{\mathrm{Noise}}
\newcommand{\pcol}{p_\mathrm{col}}
\newcommand{\softO}{\tilde{O}}
\newcommand{\softOmega}{\tilde{\Omega}}

\newcommand{\MSE}{\mathrm{MSE}}

\newcommand{\AAA}{\mathcal{A}}

\newcommand{\CC}{\mathcal{C}}

\newcommand{\HH}{\mathcal{H}}

\newcommand{\MM}{\mathcal{M}}

\newcommand{\OO}{\mathcal{O}}
\newcommand{\PP}{\mathcal{P}}

\newcommand{\RR}{\mathcal{R}}

\newcommand{\TT}{\mathcal{T}}

\newcommand{\XX}{\mathcal{X}}
\newcommand{\YY}{\mathcal{Y}}

\SetKwInput{KwPublic}{Public Parameters}

%% file: main.bbl

\begin{thebibliography}{41}


\ifx \showCODEN    \undefined \def \showCODEN     #1{\unskip}     \fi
\ifx \showDOI      \undefined \def \showDOI       #1{#1}\fi
\ifx \showISBNx    \undefined \def \showISBNx     #1{\unskip}     \fi
\ifx \showISBNxiii \undefined \def \showISBNxiii  #1{\unskip}     \fi
\ifx \showISSN     \undefined \def \showISSN      #1{\unskip}     \fi
\ifx \showLCCN     \undefined \def \showLCCN      #1{\unskip}     \fi
\ifx \shownote     \undefined \def \shownote      #1{#1}          \fi
\ifx \showarticletitle \undefined \def \showarticletitle #1{#1}   \fi
\ifx \showURL      \undefined \def \showURL       {\relax}        \fi
\providecommand\bibfield[2]{#2}
\providecommand\bibinfo[2]{#2}
\providecommand\natexlab[1]{#1}
\providecommand\showeprint[2][]{arXiv:#2}

\bibitem[\protect\citeauthoryear{{Apple Differential Privacy Team}}{{Apple
  Differential Privacy Team}}{[n.d.]}]%
        {apple:overview}
\bibfield{author}{\bibinfo{person}{{Apple Differential Privacy Team}}.}
  \bibinfo{year}{[n.d.]}\natexlab{}.
\newblock \bibinfo{title}{Apple Differential Privacy Technical Overview}.
\newblock
  \bibinfo{howpublished}{\url{https://www.apple.com/privacy/docs/Differential_Privacy_Overview.pdf}}.
\newblock


\bibitem[\protect\citeauthoryear{Balcer and Cheu}{Balcer and Cheu}{2020}]%
        {balcer2019separating}
\bibfield{author}{\bibinfo{person}{Victor Balcer} {and} \bibinfo{person}{Albert
  Cheu}.} \bibinfo{year}{2020}\natexlab{}.
\newblock \showarticletitle{Separating Local {\&} Shuffled Differential Privacy
  via Histograms}. In \bibinfo{booktitle}{\emph{1st Conference on
  Information-Theoretic Cryptography}}. \bibinfo{pages}{1:1--1:14}.
\newblock


\bibitem[\protect\citeauthoryear{Balcer, Cheu, Joseph, and Mao}{Balcer
  et~al\mbox{.}}{2021}]%
        {balcer2021connecting}
\bibfield{author}{\bibinfo{person}{Victor Balcer}, \bibinfo{person}{Albert
  Cheu}, \bibinfo{person}{Matthew Joseph}, {and} \bibinfo{person}{Jieming
  Mao}.} \bibinfo{year}{2021}\natexlab{}.
\newblock \showarticletitle{Connecting Robust Shuffle Privacy and Pan-Privacy}.
  In \bibinfo{booktitle}{\emph{Proceedings of the Thirty-Second Annual ACM-SIAM
  Symposium on Discrete Algorithms}}. \bibinfo{pages}{2384–2403}.
\newblock


\bibitem[\protect\citeauthoryear{Balle, Bell, Gasc{\'o}n, and Nissim}{Balle
  et~al\mbox{.}}{2019}]%
        {borja2019the}
\bibfield{author}{\bibinfo{person}{Borja Balle}, \bibinfo{person}{James Bell},
  \bibinfo{person}{Adri{\`a} Gasc{\'o}n}, {and} \bibinfo{person}{Kobbi
  Nissim}.} \bibinfo{year}{2019}\natexlab{}.
\newblock \showarticletitle{The Privacy Blanket of the Shuffle Model}. In
  \bibinfo{booktitle}{\emph{CRYPTO}}.
\newblock


\bibitem[\protect\citeauthoryear{Balle, Bell, Gasc\'{o}n, and Nissim}{Balle
  et~al\mbox{.}}{2020}]%
        {balle20:sum}
\bibfield{author}{\bibinfo{person}{Borja Balle}, \bibinfo{person}{James Bell},
  \bibinfo{person}{Adri\`{a} Gasc\'{o}n}, {and} \bibinfo{person}{Kobbi
  Nissim}.} \bibinfo{year}{2020}\natexlab{}.
\newblock \showarticletitle{Private Summation in the Multi-Message Shuffle
  Model}. In \bibinfo{booktitle}{\emph{ACM SIGSAC Conference on Computer and
  Communications Security}}.
\newblock


\bibitem[\protect\citeauthoryear{Bassily, Nissim, Stemmer, and
  Thakurta}{Bassily et~al\mbox{.}}{2017}]%
        {bassily2017practically}
\bibfield{author}{\bibinfo{person}{Raef Bassily}, \bibinfo{person}{Kobbi
  Nissim}, \bibinfo{person}{Uri Stemmer}, {and} \bibinfo{person}{Abhradeep
  Thakurta}.} \bibinfo{year}{2017}\natexlab{}.
\newblock \showarticletitle{Practical Locally Private Heavy Hitters}. In
  \bibinfo{booktitle}{\emph{Proceedings of the 31st International Conference on
  Neural Information Processing Systems}}. \bibinfo{pages}{2285–2293}.
\newblock


\bibitem[\protect\citeauthoryear{Bassily and Smith}{Bassily and Smith}{2015}]%
        {bassily2015local}
\bibfield{author}{\bibinfo{person}{Raef Bassily} {and} \bibinfo{person}{Adam
  Smith}.} \bibinfo{year}{2015}\natexlab{}.
\newblock \showarticletitle{Local, private, efficient protocols for succinct
  histograms}. In \bibinfo{booktitle}{\emph{{ACM} {STOC}}}.
  \bibinfo{pages}{127--135}.
\newblock


\bibitem[\protect\citeauthoryear{Beimel, Haitner, Nissim, and Stemmer}{Beimel
  et~al\mbox{.}}{2020}]%
        {amos2020on}
\bibfield{author}{\bibinfo{person}{Amos Beimel}, \bibinfo{person}{Iftach
  Haitner}, \bibinfo{person}{Kobbi Nissim}, {and} \bibinfo{person}{Uri
  Stemmer}.} \bibinfo{year}{2020}\natexlab{}.
\newblock \showarticletitle{On the Round Complexity of the Shuffle Model}. In
  \bibinfo{booktitle}{\emph{TCC 2020}}.
\newblock
\urldef\tempurl%
\url{https://arxiv.org/abs/2009.13510}
\showURL{%
\tempurl}


\bibitem[\protect\citeauthoryear{Biswas, Dong, Kamath, and Ullman}{Biswas
  et~al\mbox{.}}{2020}]%
        {biswas2020coinpress}
\bibfield{author}{\bibinfo{person}{Sourav Biswas}, \bibinfo{person}{Yihe Dong},
  \bibinfo{person}{Gautam Kamath}, {and} \bibinfo{person}{Jonathan Ullman}.}
  \bibinfo{year}{2020}\natexlab{}.
\newblock \showarticletitle{{CoinPress}: Practical Private Mean and Covariance
  Estimation}. In \bibinfo{booktitle}{\emph{Advances in Neural and Information
  Processing Systems}}.
\newblock


\bibitem[\protect\citeauthoryear{Bittau, Erlingsson, Maniatis, Mironov,
  Raghunathan, Lie, Rudominer, Kode, Tinnes, and Seefeld}{Bittau
  et~al\mbox{.}}{2017}]%
        {Bittau2017prochlo}
\bibfield{author}{\bibinfo{person}{Andrea Bittau}, \bibinfo{person}{\'{U}lfar
  Erlingsson}, \bibinfo{person}{Petros Maniatis}, \bibinfo{person}{Ilya
  Mironov}, \bibinfo{person}{Ananth Raghunathan}, \bibinfo{person}{David Lie},
  \bibinfo{person}{Mitch Rudominer}, \bibinfo{person}{Ushasree Kode},
  \bibinfo{person}{Julien Tinnes}, {and} \bibinfo{person}{Bernhard Seefeld}.}
  \bibinfo{year}{2017}\natexlab{}.
\newblock \showarticletitle{Prochlo: Strong Privacy for Analytics in the
  Crowd}. In \bibinfo{booktitle}{\emph{Proceedings of the 26th Symposium on
  Operating Systems Principles}}. \bibinfo{pages}{441–459}.
\newblock


\bibitem[\protect\citeauthoryear{Bun, Nissim, and Stemmer}{Bun
  et~al\mbox{.}}{2016}]%
        {bun2016simultaneous}
\bibfield{author}{\bibinfo{person}{Mark Bun}, \bibinfo{person}{Kobbi Nissim},
  {and} \bibinfo{person}{Uri Stemmer}.} \bibinfo{year}{2016}\natexlab{}.
\newblock \showarticletitle{Simultaneous Private Learning of Multiple
  Concepts}. In \bibinfo{booktitle}{\emph{Proceedings of the 2016 ACM
  Conference on Innovations in Theoretical Computer Science}}.
  \bibinfo{pages}{369–380}.
\newblock


\bibitem[\protect\citeauthoryear{Chan, Shi, and Song}{Chan
  et~al\mbox{.}}{2012}]%
        {chan2012optimal}
\bibfield{author}{\bibinfo{person}{T-H.~Hubert Chan}, \bibinfo{person}{Elaine
  Shi}, {and} \bibinfo{person}{Dawn Song}.} \bibinfo{year}{2012}\natexlab{}.
\newblock \showarticletitle{Optimal Lower Bound for Differentially Private
  Multi-Party Aggregation}. In \bibinfo{booktitle}{\emph{Proceedings of the
  20th Annual European Conference on Algorithms}}. \bibinfo{pages}{277–288}.
\newblock


\bibitem[\protect\citeauthoryear{Chaum}{Chaum}{1981}]%
        {chaum1982untraceable}
\bibfield{author}{\bibinfo{person}{David~L. Chaum}.}
  \bibinfo{year}{1981}\natexlab{}.
\newblock \showarticletitle{Untraceable Electronic Mail, Return Addresses, and
  Digital Pseudonyms}.
\newblock \bibinfo{journal}{\emph{Commun. ACM}}  \bibinfo{volume}{24}
  (\bibinfo{year}{1981}), \bibinfo{pages}{84–90}.
\newblock


\bibitem[\protect\citeauthoryear{Chen, Ghazi, Kumar, and Manurangsi}{Chen
  et~al\mbox{.}}{2021}]%
        {chen21distinct}
\bibfield{author}{\bibinfo{person}{Lijie Chen}, \bibinfo{person}{Badih Ghazi},
  \bibinfo{person}{Ravi Kumar}, {and} \bibinfo{person}{Pasin Manurangsi}.}
  \bibinfo{year}{2021}\natexlab{}.
\newblock \showarticletitle{On Distributed Differential Privacy and Counting
  Distinct Elements}. In \bibinfo{booktitle}{\emph{ITCS}}.
  \bibinfo{pages}{56:1--56:18}.
\newblock


\bibitem[\protect\citeauthoryear{Cheu, Smith, Ullman, Zeber, and Zhilyaev}{Cheu
  et~al\mbox{.}}{2019}]%
        {cheu2019distributed}
\bibfield{author}{\bibinfo{person}{Albert Cheu}, \bibinfo{person}{Adam Smith},
  \bibinfo{person}{Jonathan Ullman}, \bibinfo{person}{David Zeber}, {and}
  \bibinfo{person}{Maxim Zhilyaev}.} \bibinfo{year}{2019}\natexlab{}.
\newblock \showarticletitle{Distributed differential privacy via shuffling}. In
  \bibinfo{booktitle}{\emph{Annual International Conference on the Theory and
  Applications of Cryptographic Techniques}}. \bibinfo{pages}{375--403}.
\newblock


\bibitem[\protect\citeauthoryear{Cheu and Zhilyaev}{Cheu and Zhilyaev}{2022}]%
        {albert2021privatehistogram}
\bibfield{author}{\bibinfo{person}{Albert Cheu} {and} \bibinfo{person}{Maxim
  Zhilyaev}.} \bibinfo{year}{2022}\natexlab{}.
\newblock \showarticletitle{Differentially Private Histograms in the Shuffle
  Model from Fake Users}. In \bibinfo{booktitle}{\emph{2022 IEEE Symposium on
  Security and Privacy (SP)}}. \bibinfo{pages}{440--457}.
\newblock


\bibitem[\protect\citeauthoryear{Cormode and Hadjieleftheriou}{Cormode and
  Hadjieleftheriou}{2010}]%
        {cormode10frequent}
\bibfield{author}{\bibinfo{person}{G. Cormode} {and} \bibinfo{person}{M.
  Hadjieleftheriou}.} \bibinfo{year}{2010}\natexlab{}.
\newblock \showarticletitle{Methods for finding frequent items in data
  streams}.
\newblock \bibinfo{journal}{\emph{The VLDB Journal}} \bibinfo{volume}{19},
  \bibinfo{number}{1} (\bibinfo{year}{2010}), \bibinfo{pages}{3--20}.
\newblock


\bibitem[\protect\citeauthoryear{Damg\r{a}rd, Nielsen, Ostrovsky, and
  Ros\'{e}n}{Damg\r{a}rd et~al\mbox{.}}{2016}]%
        {10.5555/3081738.3081753}
\bibfield{author}{\bibinfo{person}{Ivan Damg\r{a}rd},
  \bibinfo{person}{Jesper~Buus Nielsen}, \bibinfo{person}{Rafail Ostrovsky},
  {and} \bibinfo{person}{Adi Ros\'{e}n}.} \bibinfo{year}{2016}\natexlab{}.
\newblock \showarticletitle{Unconditionally Secure Computation with Reduced
  Interaction}. In \bibinfo{booktitle}{\emph{Proceedings, Part II, of the 35th
  Annual International Conference on Advances in Cryptology --- EUROCRYPT 2016
  - Volume 9666}}. \bibinfo{publisher}{Springer-Verlag},
  \bibinfo{address}{Berlin, Heidelberg}, \bibinfo{pages}{420–447}.
\newblock
\showISBNx{9783662498958}


\bibitem[\protect\citeauthoryear{Danezis, Dingledine, and Mathewson}{Danezis
  et~al\mbox{.}}{2003}]%
        {danezis2003Mixminion}
\bibfield{author}{\bibinfo{person}{G. Danezis}, \bibinfo{person}{R.
  Dingledine}, {and} \bibinfo{person}{N. Mathewson}.}
  \bibinfo{year}{2003}\natexlab{}.
\newblock \showarticletitle{Mixminion: design of a type III anonymous remailer
  protocol}. In \bibinfo{booktitle}{\emph{Symposium on Security and Privacy.}}
  \bibinfo{pages}{2--15}.
\newblock


\bibitem[\protect\citeauthoryear{Dingledine, Mathewson, and
  Syverson}{Dingledine et~al\mbox{.}}{2004}]%
        {roger2004tor}
\bibfield{author}{\bibinfo{person}{Roger Dingledine}, \bibinfo{person}{Nick
  Mathewson}, {and} \bibinfo{person}{Paul Syverson}.}
  \bibinfo{year}{2004}\natexlab{}.
\newblock \showarticletitle{Tor: The {Second-Generation} Onion Router}. In
  \bibinfo{booktitle}{\emph{13th USENIX Security Symposium (USENIX Security
  04)}}.
\newblock


\bibitem[\protect\citeauthoryear{Erlingsson, Feldman, Mironov, Raghunathan,
  Talwar, and Thakurta}{Erlingsson et~al\mbox{.}}{2019}]%
        {erlingsson2019amplification}
\bibfield{author}{\bibinfo{person}{{\'U}lfar Erlingsson},
  \bibinfo{person}{Vitaly Feldman}, \bibinfo{person}{Ilya Mironov},
  \bibinfo{person}{Ananth Raghunathan}, \bibinfo{person}{Kunal Talwar}, {and}
  \bibinfo{person}{Abhradeep Thakurta}.} \bibinfo{year}{2019}\natexlab{}.
\newblock \showarticletitle{Amplification by shuffling: From local to central
  differential privacy via anonymity}. In \bibinfo{booktitle}{\emph{Proceedings
  of the Thirtieth Annual ACM-SIAM Symposium on Discrete Algorithms}}. SIAM,
  \bibinfo{pages}{2468--2479}.
\newblock


\bibitem[\protect\citeauthoryear{Erlingsson, Pihur, and Korolova}{Erlingsson
  et~al\mbox{.}}{2014}]%
        {erlingsson2014rappor}
\bibfield{author}{\bibinfo{person}{{\'{U}}lfar Erlingsson},
  \bibinfo{person}{Vasyl Pihur}, {and} \bibinfo{person}{Aleksandra Korolova}.}
  \bibinfo{year}{2014}\natexlab{}.
\newblock \showarticletitle{{RAPPOR:} Randomized Aggregatable
  Privacy-Preserving Ordinal Response}. In \bibinfo{booktitle}{\emph{{ACM}
  {SIGSAC} Conference on Computer and Communications Security}}.
  \bibinfo{publisher}{{ACM}}, \bibinfo{pages}{1054--1067}.
\newblock


\bibitem[\protect\citeauthoryear{Feldman and Talwar}{Feldman and
  Talwar}{2021}]%
        {feldman2021lossless}
\bibfield{author}{\bibinfo{person}{Vitaly Feldman} {and} \bibinfo{person}{Kunal
  Talwar}.} \bibinfo{year}{2021}\natexlab{}.
\newblock \showarticletitle{Lossless Compression of Efficient Private Local
  Randomizers}. In \bibinfo{booktitle}{\emph{International Conference on
  Machine Learning}}.
\newblock


\bibitem[\protect\citeauthoryear{Fitzpatrick and DeSalvo}{Fitzpatrick and
  DeSalvo}{[n.d.]}]%
        {google:COVID}
\bibfield{author}{\bibinfo{person}{Jen Fitzpatrick} {and}
  \bibinfo{person}{Karen DeSalvo}.} \bibinfo{year}{[n.d.]}\natexlab{}.
\newblock \bibinfo{title}{Helping public health officials combat COVID-19}.
\newblock
  \bibinfo{howpublished}{\url{https://blog.google/technology/health/covid-19-community-mobility-reports/}}.
\newblock


\bibitem[\protect\citeauthoryear{Ghazi, Golowich, Kumar, Pagh, and
  Velingker}{Ghazi et~al\mbox{.}}{2021a}]%
        {badih2021on}
\bibfield{author}{\bibinfo{person}{Badih Ghazi}, \bibinfo{person}{Noah
  Golowich}, \bibinfo{person}{Ravi Kumar}, \bibinfo{person}{Rasmus Pagh}, {and}
  \bibinfo{person}{Ameya Velingker}.} \bibinfo{year}{2021}\natexlab{a}.
\newblock \showarticletitle{On the Power of Multiple Anonymous Messages:
  Frequency Estimation and Selection in the Shuffle Model of Differential
  Privacy}. In \bibinfo{booktitle}{\emph{Advances in Cryptology -- EUROCRYPT
  2021}}. \bibinfo{pages}{463--488}.
\newblock


\bibitem[\protect\citeauthoryear{Ghazi, Kumar, Manurangsi, and Pagh}{Ghazi
  et~al\mbox{.}}{2020}]%
        {ghazi2020privatecounting}
\bibfield{author}{\bibinfo{person}{Badih Ghazi}, \bibinfo{person}{Ravi Kumar},
  \bibinfo{person}{Pasin Manurangsi}, {and} \bibinfo{person}{Rasmus Pagh}.}
  \bibinfo{year}{2020}\natexlab{}.
\newblock \showarticletitle{Private Counting from Anonymous Messages:
  Near-Optimal Accuracy with Vanishing Communication Overhead}. In
  \bibinfo{booktitle}{\emph{Proceedings of the 37th International Conference on
  Machine Learning}}.
\newblock


\bibitem[\protect\citeauthoryear{Ghazi, Kumar, Manurangsi, Pagh, and
  Sinha}{Ghazi et~al\mbox{.}}{2021b}]%
        {ghazi21aggregation}
\bibfield{author}{\bibinfo{person}{Badih Ghazi}, \bibinfo{person}{Ravi Kumar},
  \bibinfo{person}{Pasin Manurangsi}, \bibinfo{person}{Rasmus Pagh}, {and}
  \bibinfo{person}{Amer Sinha}.} \bibinfo{year}{2021}\natexlab{b}.
\newblock \showarticletitle{Differentially Private Aggregation in the Shuffle
  Model: Almost Central Accuracy in Almost a Single Message}. In
  \bibinfo{booktitle}{\emph{International Conference on Machine Learning}}.
\newblock


\bibitem[\protect\citeauthoryear{He, Machanavajjhala, Flynn, and Srivastava}{He
  et~al\mbox{.}}{2017}]%
        {he2017composing}
\bibfield{author}{\bibinfo{person}{Xi He}, \bibinfo{person}{Ashwin
  Machanavajjhala}, \bibinfo{person}{Cheryl Flynn}, {and}
  \bibinfo{person}{Divesh Srivastava}.} \bibinfo{year}{2017}\natexlab{}.
\newblock \showarticletitle{Composing differential privacy and secure
  computation: A case study on scaling private record linkage}. In
  \bibinfo{booktitle}{\emph{Proceedings of the 2017 ACM SIGSAC Conference on
  Computer and Communications Security}}. \bibinfo{pages}{1389--1406}.
\newblock


\bibitem[\protect\citeauthoryear{Huang, Liang, and Yi}{Huang
  et~al\mbox{.}}{2021}]%
        {huang2021instance}
\bibfield{author}{\bibinfo{person}{Ziyue Huang}, \bibinfo{person}{Yuting
  Liang}, {and} \bibinfo{person}{Ke Yi}.} \bibinfo{year}{2021}\natexlab{}.
\newblock \showarticletitle{Instance-optimal Mean Estimation Under Differential
  Privacy}. In \bibinfo{booktitle}{\emph{Advances in Neural Information
  Processing Systems}}.
\newblock


\bibitem[\protect\citeauthoryear{Ishai, Kushilevitz, Ostrovsky, and
  Sahai}{Ishai et~al\mbox{.}}{2006}]%
        {Ishai06}
\bibfield{author}{\bibinfo{person}{Y. Ishai}, \bibinfo{person}{E. Kushilevitz},
  \bibinfo{person}{R. Ostrovsky}, {and} \bibinfo{person}{A. Sahai}.}
  \bibinfo{year}{2006}\natexlab{}.
\newblock \showarticletitle{Cryptography from Anonymity}. In
  \bibinfo{booktitle}{\emph{FOCS}}.
\newblock


\bibitem[\protect\citeauthoryear{Kairouz, Oh, and Viswanath}{Kairouz
  et~al\mbox{.}}{2015}]%
        {kairouz2015secure}
\bibfield{author}{\bibinfo{person}{Peter Kairouz}, \bibinfo{person}{Sewoong
  Oh}, {and} \bibinfo{person}{Pramod Viswanath}.}
  \bibinfo{year}{2015}\natexlab{}.
\newblock \showarticletitle{Secure multi-party differential privacy}.
\newblock \bibinfo{journal}{\emph{Advances in neural information processing
  systems}}  \bibinfo{volume}{28} (\bibinfo{year}{2015}).
\newblock


\bibitem[\protect\citeauthoryear{Kopp}{Kopp}{[n.d.]}]%
        {microsoft:smartnosise}
\bibfield{author}{\bibinfo{person}{Andreas Kopp}.}
  \bibinfo{year}{[n.d.]}\natexlab{}.
\newblock \bibinfo{title}{Microsoft SmartNoise Differential Privacy Machine
  Learning Case Studies}.
\newblock
  \bibinfo{howpublished}{\url{https://azure.microsoft.com/en-in/resources/microsoft-smartnoisedifferential-privacy-machine-learning-case-studies/}}.
\newblock


\bibitem[\protect\citeauthoryear{Pass, Chowdhury, and Torgeson}{Pass
  et~al\mbox{.}}{2006}]%
        {pass2006picture}
\bibfield{author}{\bibinfo{person}{Greg Pass}, \bibinfo{person}{Abdur
  Chowdhury}, {and} \bibinfo{person}{Cayley Torgeson}.}
  \bibinfo{year}{2006}\natexlab{}.
\newblock \showarticletitle{A picture of search}. In
  \bibinfo{booktitle}{\emph{Proceedings of the 1st international conference on
  Scalable information systems}}.
\newblock


\bibitem[\protect\citeauthoryear{Pettai and Laud}{Pettai and Laud}{2015}]%
        {pettai2015combining}
\bibfield{author}{\bibinfo{person}{Martin Pettai} {and} \bibinfo{person}{Peeter
  Laud}.} \bibinfo{year}{2015}\natexlab{}.
\newblock \showarticletitle{Combining differential privacy and secure
  multiparty computation}. In \bibinfo{booktitle}{\emph{Proceedings of the 31st
  Annual Computer Security Applications Conference}}.
  \bibinfo{pages}{421--430}.
\newblock


\bibitem[\protect\citeauthoryear{Reed, Syverson, and Goldschlag}{Reed
  et~al\mbox{.}}{1998}]%
        {reed1998onion}
\bibfield{author}{\bibinfo{person}{M.G. Reed}, \bibinfo{person}{P.F. Syverson},
  {and} \bibinfo{person}{D.M. Goldschlag}.} \bibinfo{year}{1998}\natexlab{}.
\newblock \showarticletitle{Anonymous connections and onion routing}.
\newblock \bibinfo{journal}{\emph{IEEE Journal on Selected Areas in
  Communications}} (\bibinfo{year}{1998}), \bibinfo{pages}{482--494}.
\newblock


\bibitem[\protect\citeauthoryear{Reiter and Rubin}{Reiter and Rubin}{1998}]%
        {reiter1998crowds}
\bibfield{author}{\bibinfo{person}{Michael~K. Reiter} {and}
  \bibinfo{person}{Aviel~D. Rubin}.} \bibinfo{year}{1998}\natexlab{}.
\newblock \showarticletitle{Crowds: Anonymity for Web Transactions}.
\newblock \bibinfo{journal}{\emph{ACM Transactions on Information and System
  Security}} (\bibinfo{year}{1998}), \bibinfo{pages}{66–92}.
\newblock


\bibitem[\protect\citeauthoryear{Shi, Chan, Rieffel, and Song}{Shi
  et~al\mbox{.}}{2017}]%
        {shi2017distributed}
\bibfield{author}{\bibinfo{person}{Elaine Shi}, \bibinfo{person}{T.-H.~Hubert
  Chan}, \bibinfo{person}{Eleanor Rieffel}, {and} \bibinfo{person}{Dawn Song}.}
  \bibinfo{year}{2017}\natexlab{}.
\newblock \showarticletitle{Distributed Private Data Analysis: Lower Bounds and
  Practical Constructions}.
\newblock \bibinfo{journal}{\emph{ACM Transactions on Algorithms}}
  (\bibinfo{year}{2017}).
\newblock


\bibitem[\protect\citeauthoryear{Vadhan}{Vadhan}{2017}]%
        {vadhansurvey}
\bibfield{author}{\bibinfo{person}{Salil Vadhan}.}
  \bibinfo{year}{2017}\natexlab{}.
\newblock \bibinfo{booktitle}{\emph{The Complexity of Differential Privacy}}.
\newblock \bibinfo{pages}{347--450}.
\newblock


\bibitem[\protect\citeauthoryear{Wang, Blocki, Li, and Jha}{Wang
  et~al\mbox{.}}{2017}]%
        {wang2017locally}
\bibfield{author}{\bibinfo{person}{Tianhao Wang}, \bibinfo{person}{Jeremiah
  Blocki}, \bibinfo{person}{Ninghui Li}, {and} \bibinfo{person}{Somesh Jha}.}
  \bibinfo{year}{2017}\natexlab{}.
\newblock \showarticletitle{Locally differentially private protocols for
  frequency estimation}. In \bibinfo{booktitle}{\emph{26th {USENIX} Security
  Symposium Security}}. \bibinfo{pages}{729--745}.
\newblock


\bibitem[\protect\citeauthoryear{Zhou, Wang, Chan, Fanti, and Shi}{Zhou
  et~al\mbox{.}}{2022}]%
        {zhou2022locally}
\bibfield{author}{\bibinfo{person}{Mingxun Zhou}, \bibinfo{person}{Tianhao
  Wang}, \bibinfo{person}{T-H.~Hubert Chan}, \bibinfo{person}{Giulia Fanti},
  {and} \bibinfo{person}{Elaine Shi}.} \bibinfo{year}{2022}\natexlab{}.
\newblock \showarticletitle{Locally Differentially Private Sparse Vector
  Aggregation}. In \bibinfo{booktitle}{\emph{2022 IEEE Symposium on Security
  and Privacy (SP)}}. \bibinfo{pages}{422--439}.
\newblock


\bibitem[\protect\citeauthoryear{Zhu, Kairouz, McMahan, Sun, and Li}{Zhu
  et~al\mbox{.}}{2020}]%
        {wennan2020federated}
\bibfield{author}{\bibinfo{person}{Wennan Zhu}, \bibinfo{person}{Peter
  Kairouz}, \bibinfo{person}{Brendan McMahan}, \bibinfo{person}{Haicheng Sun},
  {and} \bibinfo{person}{Vivian~(Wei) Li}.} \bibinfo{year}{2020}\natexlab{}.
\newblock \showarticletitle{Federated Heavy Hitters with Differential Privacy}.
  In \bibinfo{booktitle}{\emph{International Conference on Artificial
  Intelligence and Statistics (AISTATS) 2020}}.
\newblock


\end{thebibliography}
